\documentclass[12pt, draftclsnofoot, onecolumn]{IEEEtran}
\IEEEoverridecommandlockouts
\usepackage{cite}
\usepackage{amsmath,amssymb,amsfonts}
\usepackage{cuted}
\usepackage{graphicx}
\usepackage{textcomp}
\usepackage{xcolor}
\usepackage{verbatim}
\usepackage{float}
\usepackage{bm}
\usepackage{geometry}
\usepackage{setspace}
\usepackage{stfloats}
\usepackage{multirow}
\usepackage{tabularx}
\usepackage{longtable}
\usepackage{algorithm}
\usepackage{algorithmicx}
\usepackage{algpseudocode}
\usepackage{amsthm}
\usepackage{color}
\usepackage{dsfont}
\usepackage{bbm, dsfont}
\usepackage[tight,footnotesize]{subfigure}
\newtheorem{lemma}{Lemma}

\def\BibTeX{{\rm B\kern-.05em{\sc i\kern-.025em b}\kern-.08em
    T\kern-.1667em\lower.7ex\hbox{E}\kern-.125emX}}
\IEEEoverridecommandlockouts
\begin{document}
\title{Positioning Using Visible Light Communications: A Perspective Arcs Approach
}
\author{Zhiyu Zhu, Caili Guo, ~\IEEEmembership{Senior Member,~IEEE}, Rongzhen Bao, Mingzhe Chen,  ~\IEEEmembership{Member,~IEEE},
Walid Saad ~\IEEEmembership{Fellow,~IEEE}, and Yang Yang, ~\IEEEmembership{Member,~IEEE}

\thanks{This research was supported, in part, by BUPT Excellent Ph.D. Students Foundation under Grant CX2021111, National Natural Science Foundation of China  under Grant 61901047 and 61871047, Beijing Natural Science Foundation under Grant 4204106, Fundamental Research Funds for the Central Universities under Grant 2021RC03, Proof-of-concept project of Zhongguancun Open Laboratory under Grant 202103001, and the U.S. National Science Foundation under Grant CNS-1909372.}
\thanks{Z. Zhu, R. Bao and Y. Yang are with the Beijing Key Laboratory of Network System Architecture and Convergence, School of Information and Communication Engineering, Beijing University of Posts and Telecommunications, Beijing 100876, China (e-mail: zhuzy@bupt.edu.cn; baorongzhen@bupt.edu.cn; yangyang01@bupt.edu.cn).
}
\thanks{M. Chen is with the Department of Electrical Engineering, Princeton University, Princeton, NJ, 08544, USA (e-mail: mingzhec@princeton.edu).
}
\thanks{C. Guo is with Beijing Laboratory of Advanced Information Networks, School of Information and Communication Engineering, Beijing University of Posts and Telecommunications, Beijing 100876, China (e-mail: guocaili@bupt.edu.cn).
}
\thanks{W. Saad is with the Wireless@VT, Bradley Department of Electrical and Computer Engineering, Virginia Tech, Blacksburg, VA, 24060, USA, (email: walids@vt.edu).
}
\thanks{A preliminary version of this work \cite{V-PCA} appears in the proceeding of IEEE International Conference on Communications (ICC).}

}
\maketitle
\begin{abstract}
Visible light positioning (VLP) is an accurate indoor positioning technology that uses luminaires as transmitters. In particular, circular luminaires are a common source type for VLP, that are typically treated only as point sources for positioning, while ignoring their geometry characteristics.
In this paper, the arc feature of the circular luminaire and the coordinate information obtained via visible light communication (VLC) are jointly used for positioning, and a novel perspective arcs approach is proposed for VLC-enabled indoor positioning. The proposed approach does not rely on any inertial measurement unit and has no tilted angle limitaion at the user.
First, a VLC assisted perspective circle and arc algorithm (V-PCA) is proposed for a scenario in which a complete luminaire and an incomplete one can be captured by the user. Based on plane and solid geometry theory, the relationship between the luminaire and the user is exploited to estimate the orientation and the coordinate of the luminaire in the camera coordinate system.
Then, the pose and location of the user in the world coordinate system are obtained by single-view geometry theory. Considering the cases in which parts of VLC links are blocked, an anti-occlusion VLC assisted perspective arcs algorithm (OA-V-PA) is proposed. In OA-V-PA, an approximation method is developed to estimate the projection of the luminaire's center on the image and, then, to calculate the pose and location of the user. Simulation results show that the proposed indoor positioning algorithm can achieve a 90th percentile positioning accuracy of around 10 cm. Moreover, an experimental prototype is implemented to verify the feasibility. In the established prototype, a fused image processing method is proposed to simultaneously obtain the VLC information and the geometric information. Experimental results in the established prototype show that the average positioning accuracy is less than 5 cm for different tilted angles of the user.

\end{abstract}
\section{Introduction}\label{introduction}
With the surge of location-based services such as location tracking and  navigation, accurate indoor positioning has attracted increasing attention. However, existing indoor positioning technologies, such as WiFi, RFID, Ultra-Wide Band (UWB), and Bluetooth face the challenges of balancing the positioning accuracy and the implementation cost \cite{jiao2017visible}. Since visible light possesses strong directionality and low multipath interference \cite{Matheus2019visible, yang2019relay}, visible light positioning (VLP) with visible light communication (VLC) has emerged as a promising indoor positioning technique with low cost and high positioning accuracy \cite{surveyVLP2016}. 

\subsection{Related Works}
VLP algorithms can be grouped into several categories including: proximity \cite{proximity2010}, fingerprinting \cite{fingerprint2018, bakar2020accurate, alam2018accurate}, received signal strength (RSS) \cite{CARSSR2019bai,RSSfunction, zhou2019performance} , time of arrival (TOA) \cite{TOA2013}, angle of arrival (AOA) \cite{AOA2018, soner2021visible} and image sensing \cite{imageVLP2018, fang2017high}. Proximity is the simplest positioning technique. However, it can only provide a proximity location information, whose accuracy largely depends on the density of the transmitter distribution \cite{zhuang2018survey}. Both proximity and fingerprinting can estimate the user's location using a single luminaire.
The use of fingerprinting can achieve better positioning accuracy than proximity, however, it is susceptible to environment changes \cite{Luo2017indoor}. RSS algorithms determine the location of the user based on the power of the received signal, and hence, they depend on accurate channel models \cite{RSSfunction}. Meanwhile, TOA and AOA algorithms require complicated hardware devices to ensure the positioning accuracy \cite{surveyVLP2016}. In contrast to the aforementioned algorithms that use photodiodes (PD) to receive signals, image sensing algorithms use cameras to receive visible light signals, and they determine the location of a given user by analyzing the geometric relations between LEDs and their projections on the image plane \cite{guan2019high,bai2021computer}.

Most existing VLP algorithms achieve positioning by treating indoor luminaires as point sources in indoor scenarios. 
Hence, they typically need three or more luminaires for positioning \cite{bai2020RP3P,pergoloni2017metameric,yang2021multi}. The performance of positioning may be limited by the filed of view (FOV) of the user. These approaches neglect the geometric features that the luminaires possess, such as circle and arc features, which can also be used in positioning algorithms and, ultimately, help reduce the number of required luminaires.
Recently, the works in \cite{weakprojection2017PJ, IMU2019, wang2021arbitrarily, partIMU2019PJ} studied circular luminaires based VLP using image sensing. In particular, in \cite{weakprojection2017PJ}, the authors assumed a weak perspective projection and marked a margin point on a circular luminaire to estimate the pose and the location of the user.
In \cite{IMU2019} and \cite{ wang2021arbitrarily}, the authors proposed positioning models based on inertial measurement unit (IMU) of the camera. \cite{IMU2019} proposed a planes-intersection-line positioning scheme by approximating the projected radius of the luminaire, and \cite{wang2021arbitrarily} introduced a center detecting method by using a boundary fitting method.
The work in \cite{partIMU2019PJ} used an IMU and weak perspective projection model to estimate the pose of the user, and the authors showed that this approach improves the accuracy of positioning compared with \cite{weakprojection2017PJ} by calibrating the measurement of IMU so as to reduce the azimuth angle error. Indeed, for the weak perspective projection, a large or small tilted angle can lead to inaccurate approximation, thus resulting in positioning error. The positioning schemes in \cite{IMU2019, wang2021arbitrarily, partIMU2019PJ} relied on an extra device called the IMU which is known to have inherent measurement errors due to the influence of magnetic field \cite{IMU2019}.

Overall, existing VLP approaches based on circular luminaires typically require extra devices or assume an ideal projection model, which limits their practical applications. Therefore, there is a need for a more practical indoor positioning algorithm based on circular luminaires, which can achieve high positioning accuracy without tilted angle limitation at the user, and without using an IMU.
%

%
\subsection{Contributions}
The main contribution of this paper is a new perspective arcs method for VLP that is applied in practical application scenarios.
In particular, when a complete circular luminaire and an incomplete one are captured, a perspective circle and arc positioning algorithm assisted by VLC (V-PCA) is proposed.
Meanwhile, when two incomplete circular luminaires are captured, an anti-occlusion perspective arcs algorithm assisted by VLC (OA-V-PA) is proposed for indoor positioning.
To authors' best knowledge, this is the first VLP method for circular luminaires that can achieve accurate and practical positioning without IMU and tilted angle limitations at the user.
Our key contributions are summarized as follows:
\begin{itemize}
  \item We propose an indoor positioning system with a circular luminaires layout. In the system, a perspective circle and arc algorithm, dubbed V-PCA, is proposed, and the user equipped with a single camera can be located. In V-PCA, the geometric features consisting of a complete circle and an arc are exploited to obtain the normal vector of the luminaire first. Then, we obtain the projections of the center and a mark point of circular luminaire on the image plane based on solid geometry theory. We finally estimate the location and pose of the user based on single-view geometry theory by analyzing the normal vector and the geometric information of the center and the mark point.
  \item To further enhance the practicality of V-PCA, we propose an anti-occlusion VLC assisted perspective arcs algorithm (OA-V-PA), which can achieve efficient positioning even when some of the VLC links are blocked. In OA-V-PA, only two arcs extracted from two incomplete luminaires are required to estimate the location and pose of the user. We also show that the positioning system does not need IMU and has no tilted angle limitations at the user.
  \item We implement both simulations and practical experiments to verify the performance of V-PCA and OA-V-PA. In the practical experiments, we use a mobile phone as the receiver, and we propose a fused image processing scheme to receive the visible light signals and detect the contour of the luminaire’s projection accurately. Finally, we establish an experimental prototype which can achieve accurate and real-time positioning.
\end{itemize}

Simulation results show that the proposed algorithms can achieve a 95th percentile location accuracy of less than 10 cm, and experimental results show that the proposed algorithm can achieve an average accuracy within 5 cm.

The rest of the paper is organized as follows. Section II introduces the system model. The proposed V-PCA algorithm is presented in Section III, and the proposed OA-V-PA algorithm is detailed in Section IV. Section V introduces the implementation of the algorithm. Then, simulation and experimental results are presented and discussed in Section VI. Finally, conclusions are drawn in Section VII.

Hereinafter, we adopt the following notations: both matrices and coordinates of points are denoted by capital boldface such as $\boldsymbol{R}_{\rm{c}}^{\rm{w}}$ and $\boldsymbol{P}_i^{\rm{w}}$, and they can be inferred easily by context. Vectors are denoted by $\overrightarrow {(\cdot)}$ such as $\overrightarrow {GM}$, or small boldface such as $\boldsymbol{n}_{\text{LED}}^{\text{c}}$. In addition, we use ${(\cdot)'}$ to denote projection on the image plane. For instance, $G_e'$ is the projection of $G_e$. Lines are defined and denoted by two points, such as $MN$. In addition, the definitions of key system parameters are listed in Table \ref{symbol}.

\begin{table*}[t]
		\centering \caption{summary of notations}\label{symbol}
            \footnotesize
		\begin{tabular}
{|p{0.1\textwidth}<{\centering}|p{0.3\textwidth}<{\centering}|p{0.1\textwidth}<{\centering}|p{0.3\textwidth}<{\centering}|}
\hline
			\multicolumn{1}{|c|} {\bf{Symbol}}& {\bf{Meaning}}& {\bf{Symbol}}& {\bf{Meaning}}\\
            \hline
            \hline
            $\left(u_o,v_o\right)^{\rm{T}}$ & Pixel coordinate of $O^{\rm{i}}$ & $d_x$, $d_y$ & Physical size of each pixel \\
            \hline
            $f$ & Focal length  & $\boldsymbol{Q}$ & Elliptic parameter matrix \\
            \hline
            $\boldsymbol{P}_{i}^{\text{w}}$/$\boldsymbol{P}_{i}^{\text{c}}$/$\boldsymbol{P}_{i}^{\text{a}}$ & Coordinate of point $i$ in WCS/CCS/ACS & $\Pi_e$/$\Pi_f$ & Fitted ellipse of luminaire's projection \\
            \hline
            $M$ & Mark point of the luminaire & ${G_e}$/${G_f}$ & Center of the circular luminaire\\
            \hline
            $M'$ & Projection of $M$ on image plane & ${G'_e}$/${G'_f}$ & Projection of ${G_e}$/${G_f}$ on image plane\\
            \hline
            $\Gamma$ & Elliptical cone whose vertex is the camera and base is the fitted ellipse & $\Gamma_c$ & Conical surface that is the lateral face of elliptical cone $\Gamma$\\
            \hline
            $\varphi$, $\theta$, $\psi$ & Euler angles around $x^{\rm c}$-axis, $y^{\rm c}$-axis, $z^{\rm c}$-axis & ${\boldsymbol{n}}_{{\rm{LED}}}^{\rm{w}}$/ ${\boldsymbol{n}}_{{\rm{LED}}}^{\rm{c}}$/ ${\boldsymbol{n}}_{{\rm{LED}}}^{\rm{a}}$ & Normal vector of the luminaire in WCS/CCS/ACS \\
            \hline
            $\boldsymbol{R}_{\text{c}}^{\text{w}}$/$\boldsymbol{t}_{\text{c}}^{\text{w}}$ & Pose/location of the camera in WCS (Rotation matrix/translation vector from CCS to WCS) & $\boldsymbol{R}_{\text{c}}^{\text{a}}$ & Pose of the camera in ACS (Rotation matrix from CCS to ACS)\\
            \hline
		\end{tabular}
\end{table*}

\section{System model and Problem Formulation}\label{systemmodel}
\subsection{System Model}
We consider an indoor VLP system that consists of a circular panel luminaires layout and one user equipped with a standard pinhole camera, as shown in Fig. \ref{scenario}(a). The circular panel luminaires are mounted on the ceiling of a room and are assumed to face vertically downwards. Each luminaire consists of several LEDs, each of which transmits the VLC information to the user. The user uses its camera to receive the VLC information and capture the image of the luminaires to estimate its current pose and location. 
\begin{figure*}[tpb]
  \centering
  \subfigure[Illustration of the system scenario.] {\includegraphics[height = 4cm, trim = 20 0 0 0]{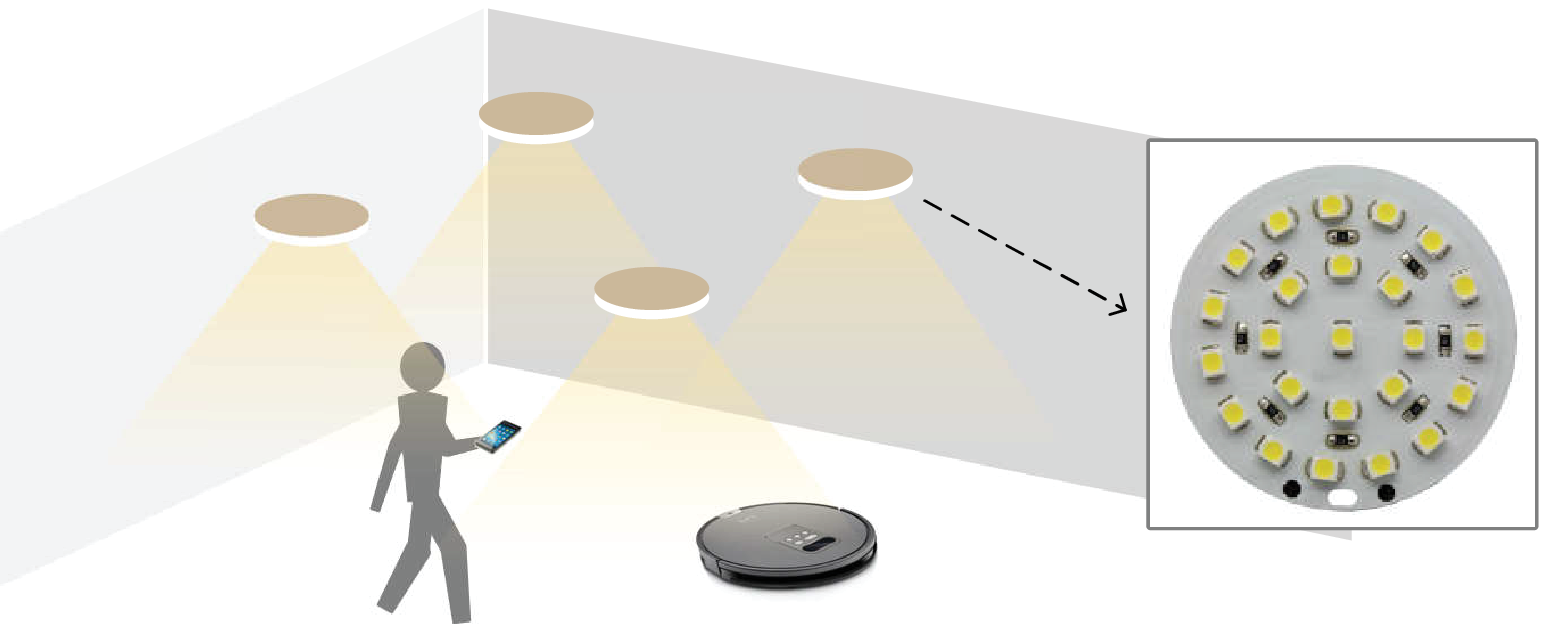}}
  \subfigure[The projection diagram of the system.] {\includegraphics[height=5.5cm,clip]{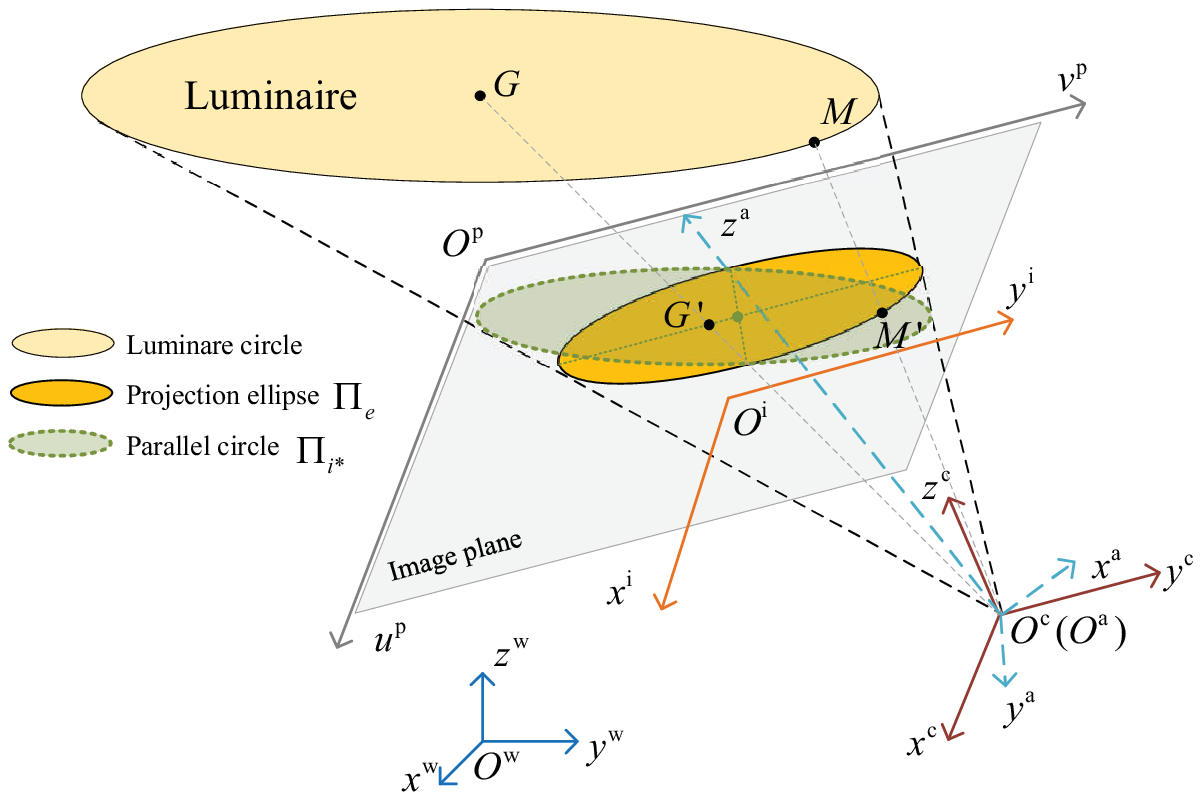}}
  \caption{System model diagrams.}\label{scenario}
\end{figure*}

In the considered system, the geometric relationship between the luminaire and the user will be used to estimate the pose and location of the user. In particular, each luminaire can be captured by the user's camera, and its projection is an ellipse on the image plane as shown in Fig. \ref{scenario}(b). 
The ellipse ${\Pi _e}$ on the image plane is the projection of a luminaire. The projection of the luminaire and the user can generate an elliptical cone $\Gamma$, whose base is ellipse ${\Pi _e}$, and vertex is ${O^{\rm{c}}}$. The lateral surface of the elliptical cone $\Gamma$ is a conical surface that is denoted by $\Gamma_c$.

To accurately estimate the user's location and pose, we need to use four coordinate systems: a) the two-dimensional (2D) pixel coordinate system (PCS) ${O^{\rm{p}}} - {u^{\rm{p}}}{v^{\rm{p}}}$ on the image plane, b) the 2D image coordinate system (ICS) ${O^{\rm{i}}} - {x^{\rm{i}}}{y^{\rm{i}}}$  on the image plane, c) the three-dimensional (3D) camera coordinate system (CCS) ${O^{\rm{c}}} - {x^{\rm{c}}}{y^{\rm{c}}}{z^{\rm{c}}}$, and d) the 3D world coordinate system (WCS) ${O^{\rm{w}}} - {x^{\rm{w}}}{y^{\rm{w}}}{z^{\rm{w}}}$.
In particular, ${O^{\rm{p}}}$, ${O^{\rm{i}}}$, ${O^{\rm{c}}}$ and ${O^{\rm{w}}}$ are the origins of PCS, ICS, CCS and WCS, respectively. In PCS, ICS, and CCS, the coordinate axes ${u^{\rm{p}}}$, ${x^{\rm{i}}}$, and ${x^{\rm{c}}}$ are parallel. Meanwhile, ${v^{\rm{p}}}$, ${y^{\rm{i}}}$, and ${y^{\rm{c}}}$ are parallel.
${O^{\rm{i}}}$ and ${O^{\rm{c}}}$ are on the optical axis, and the distance between them is the focal length $f$. Therefore, the $z$-coordinate of the image plane in CCS is ${z^{\rm{c}}} = f$. Here, we term the coordinate of a point in PCS/ICS/CCS/WCS as pixel coordinate/image coordinate/camera coordinate/world coordinate. Next, we summarize the relationships among the four coordinate systems in the following lemmas.

\begin{lemma}\label{PCStoICS}
\rm{On the image plane, given a pixel coordinate, ${\boldsymbol{P}}_{i'}^{\rm{p}} = (u_{i'}^{\rm{p}},v_{i'}^{\rm{p}})^{\rm{T}}$, of point $i'$ in PCS,
the image coordinate, ${\boldsymbol{P}}_{i'}^{\rm{i}}=(x_{i'}^{\rm{i}},y_{i'}^{\rm{i}})^{\rm{T}}$, of point $i'$ in ICS can be obtained by

\begin{equation}\label{PCStoICS}
{\boldsymbol{P}}_{i'}^{\text{i}}=\left[ \begin{matrix}
   x_{i'}^{\text{i}} \\
  y_{i'}^{\text{i}} \\
\end{matrix} \right]=\left[ \begin{matrix}
   {{d}_{x}} & 0  \\
   0 & {{d}_{y}}  \\
\end{matrix} \right]\cdot {\boldsymbol{P}}_{i'}^{\text{p}}-\left[ \begin{matrix}
   {{u}_{0}}{{d}_{x}}  \\
   {{v}_{0}}{{d}_{y}}  \\
\end{matrix} \right],
\end{equation}
where ${d_x}$ and ${d_y}$ are the physical size of each pixel along $x^{\rm{i}}$ and $y^{\rm{i}}$ axes on the image plane, respectively, and $\left(u_o,v_o\right)^{\rm{T}}$ is the pixel coordinate of the origin, $O^{\rm{i}}$}, of ICS.
\end{lemma}
\begin{proof}
Given the pixel coordinate of the origin, $\left(u_o,v_o\right)^{\rm{T}}$, the physical size, ${d_x}$ and ${d_y}$\footnote{The origin $\left(u_o,v_o\right)^{\rm{T}}$ and the physical sizes ${d_x}$ and ${d_y}$, are camera’s intrinsic parameters, and can be calibrated in advance \cite{CARSSR2019bai}. }, and the pixel coordinate, ${\boldsymbol{P}}_{i'}^{\rm{p}} = (u_{i'}^{\rm{p}},v_{i'}^{\rm{p}})^{\rm{T}}$, of point $i'$ on the image plane, $x_{i'}^{\rm{i}} = {d_x}\left(u_{i'}^{\rm{p}} - u_o\right)$ and $y_{i'}^{\rm{i}} = {d_y}\left(v_{i'}^{\rm{p}} - v_o\right)$ can be obtained according to single-view geometry theory \cite{zisserman1997geometric}. This completes the proof.

\end{proof}

\begin{lemma}
\rm{Given a camera coordinate, ${\boldsymbol{P}}_i^{\rm{c}} = (x_i^{\rm{c}},y_i^{\rm{c}},z_i^{\rm{c}})^{\rm{T}}$, of point $i$ in CCS, the image coordinate, ${\boldsymbol{P}}_{i'}^{\rm{i}} = (x_{i'}^{\rm{i}},y_{i'}^{\rm{i}})^{\rm{T}}$, of its projection point $i'$ on the image plane can be obtained by
\begin{equation}\label{ICStoCCS}
   {\boldsymbol{P}}_{i}^{\text{c}} =\left[ \begin{matrix}
   x_{i}^{\text{c}} \\
  y_{i}^{\text{c}} \\
  z_{i}^{\text{c}} \\
\end{matrix} \right]=z_{i}^{\text{c}}\left[ \begin{matrix}
   1/f & 0 & 0  \\
   0 & 1/f & 0  \\
   0 & 0 & 1  \\
\end{matrix} \right]\left[ \begin{matrix}
   x_{i'}^{\text{i}} \\
  y_{i'}^{\text{i}}   \\
   1  \\
\end{matrix} \right]=z_{i}^{\text{c}}\left[ \begin{matrix}
   1/f & 0 & 0  \\
   0 & 1/f & 0  \\
   0 & 0 & 1  \\
\end{matrix} \right]\left[ \begin{matrix}
   {\boldsymbol{P}}_{i'}^{\rm{i}}  \\
   1  \\
\end{matrix} \right],
\end{equation}
Particularly, the camera coordinate of point $i'$ in CCS is ${\boldsymbol{P}}_{i'}^{\rm{c}}= (x_{i'}^{\rm{i}},y_{i'}^{\rm{i}},f)^{\rm{T}}$.}
\end{lemma}

\begin{proof}
According to the triangle similarity in projection diagram, we have $\frac {{x_{i}^{\text{c}}}}{{x_{i'}^{\text{i}}}} = \frac{z_{i}^{\text{c}}}{f}$ and $\frac {{y_{i}^{\text{c}}}}{{y_{i'}^{\text{i}}}} = \frac{z_{i}^{\text{c}}}{f}$. Hence, we can obtain the coordinates of ${\boldsymbol{P}}_i^{\rm{c}}$ using ${\boldsymbol{P}}_{i'}^{\rm{i}}$. This completes the proof of (\ref{ICStoCCS}). The image plane is perpendicular to the optical axis, and thus, any point on the image plane has the $z$-coordinate of $z^{\text{c}}=f$. Then, we can have the camera coordinate of point $i'$ as ${\boldsymbol{P}}_{i'}^{\rm{c}}= (x_{i'}^{\rm{i}},y_{i'}^{\rm{i}},f)^{\rm{T}}$.
\end{proof}

\begin{lemma}
\rm{Given a camera coordinate, ${\boldsymbol{P}}_i^{\rm{c}} = (x_i^{\rm{c}},y_i^{\rm{c}},z_i^{\rm{c}})^{\rm{T}}$, of point $i$, a rotation matrix $\boldsymbol{R}_{\text{c}}^{\text{w}}$, and a translation vector $\boldsymbol{t}_{\text{c}}^{\text{w}}$, the world coordinate, ${\boldsymbol{P}}_i^{\rm{w}} = (x_i^{\rm{w}},y_i^{\rm{w}},z_i^{\rm{w}})^{\rm{T}}$ of point $i$ can be obtained by \cite{bai2021computer}
\begin{equation}\label{RTexpress}
{{\boldsymbol{P}}_{i}^{\rm{w}}} = {\boldsymbol{R}}_{\rm{c}}^{\rm{w}}  \cdot {{\boldsymbol{P}}_{i}^{\rm{c}}} + {\boldsymbol{t}}_{\rm{c}}^{\rm{w}},
\end{equation}
where $\boldsymbol{R}_{\text{c}}^{\text{w}}$ is a $3 \times 3$ matrix, which also represents the pose of the user. $\boldsymbol{t}_{\text{c}}^{\text{w}}$ is a $3 \times 1$ vector, which also represents the location of the user.}
\end{lemma}

\begin{proof}
$\boldsymbol{R}_{\text{c}}^{\text{w}}$ contains the rotations corresponding to $x^{\rm{c}}$, $y^{\rm{c}}$, and $z^{\rm{c}}$ axes in CCS. Thus,  $\boldsymbol{R}_{\text{c}}^{\text{w}}$ can be considered as the pose of the user. Since ${{\boldsymbol{P}}_{i}^{\rm{w}}} = {\boldsymbol{t}}_{\rm{c}}^{\rm{w}}$ when ${{\boldsymbol{P}}_{ i}^{\rm{c}}}={\left( {0,0,0} \right)^{\rm{T}}}$ according to (\ref{RTexpress}), $\boldsymbol{t}_{\rm{c}}^{\rm{w}}$ can be considered as the location of the user.
\end{proof}

After establishing the relationships among the four coordinate systems, we can use the geometric information on the image plane to deduce the relationship between the user and the spacial luminaire, so as to estimate the user's location and pose. This design problem will be formulated next.

\subsection{Problem Formulation}
Given the introduced system model, we now introduce our pose and location estimation problem. Our purpose is to obtain accurate $\boldsymbol R_{\rm{c}}^{\rm{w}}$ and $\boldsymbol t_{\rm{c}}^{\rm{w}}$ in (\ref{RTexpress}). Let ${{\boldsymbol{R}}_X}$, ${{\boldsymbol{R}}_Y}$, and ${{\boldsymbol{R}}_Z}$ be the rotation matrices around axis $x^{\rm c}$, $y^{\rm c}$, and $z^{\rm c}$, respectively. Then, the pose of the user can be expressed as
\begin{equation}\label{Rcw}
{\boldsymbol{R}}_{\rm{c}}^{\rm{w}} = {{\boldsymbol{R}}_X}{{\boldsymbol{R}}_Y}{{\boldsymbol{R}}_Z},
\end{equation}
where ${{\boldsymbol{R}}_X}$, ${{\boldsymbol{R}}_Y}$, and ${{\boldsymbol{R}}_Z}$ can be denoted by \cite{bai2021computer}
\begin{equation}\label{Rx}
{{\boldsymbol{R}}_X} = \left[ {\begin{array}{*{20}{c}}
1&0&0\\
0&{\cos \varphi }&{ - \sin \varphi }\\
0&{\sin \varphi }&{\cos \varphi }
\end{array}} \right],
\end{equation}
\begin{equation}\label{Ry}
{{\boldsymbol{R}}_Y} = \left[ {\begin{array}{*{20}{c}}
{\cos \theta }&0&{\sin \theta }\\
0&1&0\\
{ - \sin \theta }&0&{\cos \theta }
\end{array}} \right],
\end{equation}
and
\begin{equation}\label{Rz}
\vspace{-0.1cm}
{{\boldsymbol{R}}_Z} = \left[ {\begin{array}{*{20}{c}}
{\cos \psi }&{ - \sin \psi }&0\\
{\sin \psi }&{\cos \psi }&0\\
0&0&1
\end{array}} \right],
\end{equation}
with $\varphi$, $\theta$, and $\psi$ being the Euler angles around $x^{\rm c}$-axis, $y^{\rm c}$-axis, and $z^{\rm c}$-axis, respectively. Hence, if we want to estimate the pose $\boldsymbol R_{\rm{c}}^{\rm{w}}$ of the user, we need to estimate the Euler angles $\varphi$, $\theta$, and $\psi$. Given the estimated $\boldsymbol R_{\rm{c}}^{\rm{w}}$, the location $\boldsymbol t_{\rm{c}}^{\rm{w}}$ of the user can also be estimated according to (\ref{RTexpress}).

We consider two scenarios. In the \emph{first scenario}, we analyze a case in which a complete circular luminaire and an incomplete one can be captured by the camera. As introduced in Section \ref{introduction}, when using circular luminaires for positioning in existing studies, an IMU{\footnote{IMU is an electronic device that can measure the orientation of the user, and it usually has measurement errors \cite{IMU2019}.}} must be used to estimate the Euler angles, the accuracy of which is significantly affected by the estimation deviation of the IMU. Moreover, when circle features are employed in pose estimation, there lacks point correspondences, and there are usually dual solutions. To circumvent these challenges, the V-PCA algorithm is proposed for accurate positioning. We propose two methods including space-time coding and artificially marking to obtain point correspondences. We also propose to use the geometric features extracted from the ellipse image of a circular luminaire to estimate the Euler angles, and use another captured incomplete luminaire for disambiguating.
Then, considering that parts of VLC links may be blocked, we further consider a \emph{second scenario} in which two incomplete luminaires are captured. However, the point correspondences in V-PCA are no longer appropriate for this scenario. Hence, we propose the OA-V-PA algorithm with an approximation method developed to solve point correspondences. In the following sections, we detail the proposed V-PCA and OA-V-PA algorithms.

\section{V-PCA Positioning Algorithm}\label{V-PCA}
In this section, we consider a scenario in which a complete luminaire and an incomplete luminaire are simultaneously captured by the camera, and then, a novel positioning algorithm, called V-PCA, is introduced for the user's pose and location estimation.
Compared to existing VLP algorithms that require IMU for Euler angles estimation, the proposed algorithm can achieve higher positioning accuracy with only image sensor, thereby avoiding measurement errors. The proposed V-PCA algorithm consists of four steps:
a) Estimate the normal vector ${\boldsymbol{n}}_{{\rm{LED}}}^{\rm{c}}$ of the luminaire in CCS.
b) Estimate the coordinates of the center $G_e$ and the mark point $M$ in CCS, i.e., ${\boldsymbol {P}}_{G_e}^{\rm{c}}$ and ${\boldsymbol {P}}_{M}^{\rm{c}}$.
c) Estimate the pose and the location of the user in WCS based on $ {\boldsymbol{n}}_{{\rm{LED}}}^{\rm{c}}$, ${\boldsymbol {P}}_{G_e}^{\rm{c}}$ and ${\boldsymbol {P}}_{M}^{\rm{c}}$ by using single-view geometry theory.
Next, we introduce the various steps of our proposed V-PCA algorithm.

%

\subsection{Normal Vector of Luminaires in CCS}\label{normalvector}
To estimate the pose and location of the user, the first step is to find the normal vector ${\boldsymbol{n}}_{{\rm{LED}}}^{\rm{c}}$ of the luminaire in CCS, and, thus, determine the orientation of the user. To obtain ${\boldsymbol{n}}_{{\rm{LED}}}^{\rm{c}}$, we establish an auxiliary coordinate system (ACS) and exploit the geometric relation between the luminaire and its projection on the image plane.

On the image plane, the elliptic curve ${\Pi _e}$ of a luminaire can be expressed as
\begin{equation}\label{imgellipse}
\vspace{-0.1cm}
 A{(x^{\rm{i}}})^2 + B{x^{\rm{i}}}{y^{\rm{i}}} + C({y^{\rm{i}}})^2 + D{x^{\rm{i}}} + E{y^{\rm{i}}} + {\rm{1}} = 0,
\vspace{-0.1cm}
\end{equation}
where $A$, $B$, $C$, $D$, and $E$ are parameters that can be determined by curving fitting with the least squares fitting \cite{ellipsefitting}.

\begin{lemma}
\rm{Given ICS, CCS and an expression of elliptic curve, ${\Pi _e}$, i.e., $ A{(x^{\rm{i}}})^2 + B{x^{\rm{i}}}{y^{\rm{i}}} + C({y^{\rm{i}}})^2 + D{x^{\rm{i}}} + E{y^{\rm{i}}} + {\rm{1}} = 0$, there must exist a canonical frame of conicoid, i.e., ACS, which enables the conical surface $\Gamma_c$ of the elliptical cone $\Gamma$ can be represented in a compact form as
\begin{equation}\label{TCScone2}
{\lambda _1}({x^{{\rm{a}}}})^2 + {\lambda _2}({y^{{\rm{a}}}})^2 + {\lambda _3}({z^{{\rm{a}}}})^2 = 0,
\end{equation}
in ACS. ${\lambda _1},{\lambda _2}$, and ${\lambda _3}$ are eigenvalues of ${\boldsymbol{Q}} = \left[ {\begin{array}{*{20}{c}}
{A{f^2}}  &  {B{f^2}/2}  &  {Df/2}\\
{B{f^2}/2}  &  {C{f^2}}  &  {Ef/2}\\
{Df/2}  &  {Ef/2}  &  1
\end{array}} \right]$, and the rotation matrix from ACS to CCS is denoted by ${{\boldsymbol{R}}_{{\rm{a}}}^{\rm{c}}}$ , which consists of the eigenvectors of $\boldsymbol{Q}$.}
\end{lemma}

\begin{proof}
By substituting (\ref{ICStoCCS}) into (\ref{imgellipse}), we can have a general form of the conical surface $\Gamma_c$ in CCS as
\begin{equation}\label{CCScone1}
 A{f^2}(\!{x^c}\!)^2 \! +\! B{f^2}{x^c}{y^c}\! +\! C{f^2}(\!{y^c}\!)^2 \!+\! Df{x^c}{z^c} \!+ \! Ef{y^c}{z^c} \!+ \!(\!{z^{c}}\!)^2 =\left[ \begin{matrix}
   {{x}^{\text{c}}} & {{y}^{\text{c}}} & {{z}^{\text{c}}}  \\
\end{matrix} \right]\mathbf{Q}{{\left[ \begin{matrix}
   {{x}^{\text{c}}} & {{y}^{\text{c}}} & {{z}^{\text{c}}}  \\
\end{matrix} \right]}^{\text{T}}}=0 \\
\end{equation}

where ${\boldsymbol{Q}} = \left[ {\begin{array}{*{20}{c}}
{A{f^2}}  &  {B{f^2}/2}  &  {Df/2}\\
{B{f^2}/2}  &  {C{f^2}}  &  {Ef/2}\\
{Df/2}  &  {Ef/2}  &  1
\end{array}} \right]$. Matrix $\boldsymbol{Q}$ is not a diagonal matrix since the central axis of elliptical cone $\Gamma$ does not coincide with any coordinate axes of CCS. To simplify (\ref{CCScone1}), we first assume a canonical frame of conicoid, ACS, and a rotation matrix, ${{\boldsymbol{R}}_{{\rm{a}}}^{\rm{c}}}$, from ACS to CCS. Then, the transformation between ACS and CCS can be expressed as

\begin{equation}\label{TCStoCCS}
\vspace{-0.2cm}
{\left[ {\begin{array}{*{20}{c}}
{x^{\rm{c}}}&{y^{\rm{c}}}&{z^{\rm{c}}}
\end{array}} \right]^{\rm{T}}} = {\boldsymbol{R}}_{{\rm{a}}}^{\rm{c}}{\left[ {\begin{array}{*{20}{c}}
{x^{{\rm{a}}}}&{y^{{\rm{a}}}}&{z^{{\rm{a}}}}
\end{array}} \right]^{\rm{T}}},
\end{equation}
By substituting (\ref{TCStoCCS}) into (\ref{CCScone1}), we have
\begin{equation}\label{TCScone1}
\left[ {\begin{array}{*{20}{c}}
{{x^{{\rm{a}}}}} \!  & \!  {{y^{{\rm{a}}}}} \!\!  &  \!\!{{z^{{\rm{a}}}}}
\end{array}} \right]{\left( {{\boldsymbol{R}}_{{\rm{a}}}^{\rm{c}}} \right)^{\rm{T}}}{\boldsymbol{QR}}_{{\rm{a}}}^{\rm{c}}{\left[ {\begin{array}{*{10}{c}}
{{x^{{\rm{a}}}}}  \!\!  &  \!\!  {{y^{{\rm{a}}}}}  \!\!  &  \!\!{{z^{{\rm{a}}}}}
\end{array}} \right]^{\rm{T}}} = 0.
\end{equation}
Next, by diagonalizing matrix $\bf{Q}$, (\ref{TCScone1}) can be derived as
\begin{equation}\label{diag}
  ({{\boldsymbol{R}}_{{\rm{a}}}^{\rm{c}}})^{\rm{T}}{\boldsymbol{Q{{\boldsymbol{R}}_{{\rm{a}}}^{\rm{c}}}}} = ({{{\boldsymbol{R}} _{{\rm{a}}}^{\rm{c}}})^{-1}}{\boldsymbol{Q{{\boldsymbol{R}}_{{\rm{a}}}^{\rm{c}}}}} = {\rm{diag}}\left({{\lambda _1}},{{\lambda _2}},{{\lambda _3}} \right),
\end{equation}
where ${{\boldsymbol{R}}_{{\rm{a}}}^{\rm{c}}}$ consists of the eigenvectors of $\boldsymbol{Q}$. In addition, ${\lambda _1},{\lambda _2}$, and ${\lambda _3}$ are eigenvalues of $\boldsymbol{Q}$. Hence, in ACS, a compact form equation of conical surface $\Gamma_c$ in (\ref{TCScone2}) can be derived from (\ref{TCScone1}).

Note that matrix $\boldsymbol{Q}$ is a real symmetric matrix, and thus it can be diagonalized definitely. After diagonalizing matrix $\boldsymbol{Q}$, ${{\boldsymbol{R}}_{{\rm{a}}}^{\rm{c}}}$ and $\lambda_i, i=1,2,3$ can be obtained. Meanwhile, the corresponding ACS can also be determined.
\end{proof}

ACS aligns one of the coordinate axes with the central axis of elliptical cone $\Gamma$, and it also shares the same origin with CCS. The values of ${\lambda _i}, i=1,2,3$ determine the relative position of the elliptical cone $\Gamma$ and the coordinate axes of ACS.
To apply the form of the conical surface $\Gamma_c$ in (\ref{TCScone2}), there must be one of ${\lambda _i}, i=1,2,3$ less than 0, while the other two greater than 0. In addition, there are six possible cases among ${\lambda _1},{\lambda _2}$, and ${\lambda _3}$. Each case represents a relation between the elliptical cone and the coordinate axes. For instance, ${\lambda _3} < 0 < {\lambda _2} \le {\lambda _1}$ means that the central axis of elliptical cone $\Gamma$ aligns with ${z^{\rm{a}}}$ axis, the major-axis of the ellipse is parallel to ${y^{\rm{a}}}$ axis, and the minor-axis is parallel to ${x^{\rm{a}}}$ axis. When ${\lambda _1},{\lambda _2}$, and ${\lambda _3}$ have other relations, we can always find the corresponding relations between the elliptical cone and the coordinate axes.

We estimate the normal vector of the luminaire plane by finding a parallel plane to it. In particular, rotating the plane that ellipse ${\Pi _e}$ lies on around its  major axis, we can obtain two circles, termed as ${\Pi _1}$ and ${\Pi _2}$, that intersect with conical surface $\Gamma_c$, and only one of them is parallel to the luminaire. ${\Pi _1}$ and ${\Pi _2}$ have different orientations, and both circles are parallel to ${y^{{\rm{a}}}}$ axis. Without loss of generality, we define planes ${\Pi _1}$ and ${\Pi _2}$ as
\begin{equation}\label{sectioncircle1}
{z^{{\rm{a}}}} = {k_1}{x^{{\rm{a}}}} + b,
\end{equation}
and
\begin{equation}\label{sectioncircle2}
{z^{{\rm{a}}}} = {k_2}{x^{{\rm{a}}}} + b,
\end{equation}
respectively, where ${k_1} =  \sqrt {\frac{{{\lambda _1} - {\lambda _2}}}{{{\lambda _2} - {\lambda _3}}}} $ and ${k_2} =  -\sqrt {\frac{{{\lambda _1} - {\lambda _2}}}{{{\lambda _2} - {\lambda _3}}}} $ decide the orientation of the circular planes \cite{3Dcirclelocation1992}, and $b$ is a constant. Given variables ${k_1}$ and ${k_2}$, there exists two possible normal vectors of the luminaire, and they can be expressed as ${\boldsymbol{n}}_{{\Pi _1}}^{{\rm{a}}}= {\left( {k_1}  ,0,{ - 1}  \right)^{\rm{T}}}$ and ${\boldsymbol{n}}_{{\Pi _{2}}}^{{\rm{a}}} =  {\left( {k_2} ,0,{ - 1}  \right)^{\rm{T}}}$.

To eliminate the duality of the normal vector, an arc of another LED luminaire is utilized to find the correct one, instead of requiring another complete luminaire. This arc can be fitted into another ellipse ${\Pi _f}$. Similarly, we further construct the elliptical cone using ${\Pi _f}$ and the camera vertex, and rotate the plane that ${\Pi _{f}}$ lies on to obtain two circles with the normal vectors of ${\boldsymbol{n}}_{{\Pi _3}}^{{\rm{a}}}$ and ${\boldsymbol{n}}_{{\Pi _4}}^{{\rm{a}}}$. Since both the arc and the complete circle are on the ceiling of the room, there exists two identical items among ${\boldsymbol{n}}_{{\Pi _1}}^{{\rm{a}}}$, $ {\boldsymbol{n}}_{{\Pi _2}}^{{\rm{a}}}$, $ {\boldsymbol{n}}_{{\Pi _3}}^{{\rm{a}}}$, and ${\boldsymbol{n}}_{{\Pi _4}}^{{\rm{a}}}$ in theory. Then, we can solve
\begin{equation}\label{nvjudge}
\mathop {\min }\limits_{i,j} {\left\| {{\boldsymbol{n}}_{{\Pi _{i}}}^{{\rm{a}}} - {\boldsymbol{n}}_{{\Pi _{j}}}^{{\rm{a}}}} \right\|},i = 1,2,j=3,4,
\end{equation}
to obtain the optimal plane ${\Pi _{i^*}}$ and the normal vector ${\boldsymbol{n}}_{{\Pi _{{i^*}}}}^{{\rm{a}}}$, where ${i^*}$ represents the optimal $i$  obtained from (\ref{nvjudge}). Based on ${\boldsymbol{n}}_{{\Pi _{{i^*}}}}^{{\rm{a}}}$, the normalized normal vector of the LED luminaire in CCS can be expressed as
\begin{equation}\label{nvCCS}
\boldsymbol{n}_{\text{LED}}^{\text{c}}=\boldsymbol{R}_{\text{a}}^{\text{c}}\frac{\boldsymbol{n}_{{{\Pi }_{i^*}}}^{\text{a}}}{\sqrt{k_{i^*}^{2}+1}}.
\end{equation}
As such, the normal vector of the luminaires can be obtained, and it will be used in the subsequent pose and location estimation. For convenience, we denote the obtained ${\boldsymbol{n}}_{{\rm{LED}}}^{\rm{c}}$ in (\ref{nvCCS}) as ${\boldsymbol{n}}_{{\rm{LED}}}^{\rm{c}} = {( {t_1 },{t_2 },{t_3 })}^{\rm{T}}$.

\subsection{Center and Mark Point Coordinates in CCS}\label{GcMcESTIMATION}
To achieve positioning, two points, ${G_e}$ and ${M}$, on the luminaire are required, and their coordinates in CCS, ${\boldsymbol {P}}_{G_e}^{\rm{c}}$ and ${\boldsymbol {P}}_{M}^{\rm{c}}$,  should be obtained.
In particular, ${G_e}$ is at the center of the luminaire, while $M$ is a mark point on the margin of the luminaire. The vector $\overrightarrow {G_eM}$ can be normalized as ${\boldsymbol{n}}_{G_eM}^{\rm{w}} = {\left( {0,1,0} \right)^{\rm{T}}}$ in WCS.
In this section, we leverage the geometric relationship between the points and their projections on the image plane to estimate ${\boldsymbol {P}}_{G_e}^{\rm{c}}$ and ${\boldsymbol {P}}_{M}^{\rm{c}}$.

We first propose to obtain the projections of ${G_e}$ and ${M}$ on the image plane. ${G'_e}$ and ${M'}$ are the projections of ${G_e}$ and ${M}$, respectively. There are two methods to obtain the projections proposed. The first method is to artificially mark. The points can be marked directly on the luminaire, as implemented in \cite{weakprojection2017PJ}. The second method designs a space-time coding model to obtain points ${G'_e}$ and ${M'}$ by taking advantage of VLC. One possible space-time coding model is designed as shown in Fig. \ref{spacetime}. The codes are sequentially transmitted by the luminaire.
It can be found that the projection point $G_e'$ is the intersection point of lines $P'G_e'$ and $M'G_e'$.  Lines $P'G_e'$ and $M'G_e'$ are the projections of lines $PG_e$ and $MG_e$, respectively. The projection point $M'$ can be found as the endpoint of the overlapped line, $M'G_e'$, of two sequential codes. Thus, the points $G_e'$ and $M'$ can be found on the image plane, and their coordinates can be obtained through image processing \cite{liVLCsmartphoneIPS}.
Even though in Fig. \ref{spacetime} there are only a few active LEDs, the space-time coding design will not affect the illumination level. This is because that the space-time coding can be carefully designed to achieve a target dimming level \cite{generalized}. In addition, all the LEDs on the same luminaire broadcast the luminaire's ID information and the world coordinates information of the center $G_e$ and the mark point $M$.

\begin{figure}[t]
\centering
\includegraphics[height=5cm,trim=0 0 0 34,clip]{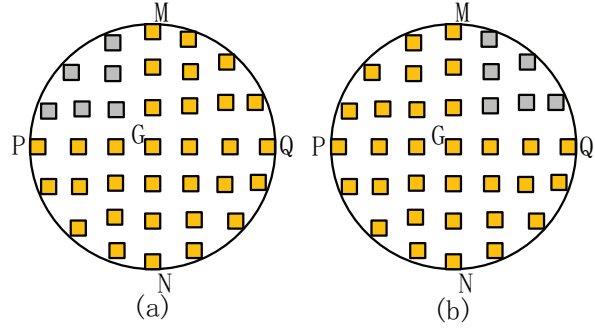}
\caption{Space-time coding model.}\label{spacetime}
\vspace{-0.5cm}
\end{figure}


Then, we exploit the relationship between the points and their projections to obtain ${\boldsymbol{P}}_{G_e}^{\rm{c}}$ and ${\boldsymbol{P}}_{M}^{\rm{c}}$. As shown in Fig. \ref{scenario}(b), the luminaire plane intersects line ${M'} O^{{\rm{c}}}$ at point $M$, while intersects line ${{G'_e} O^{{\rm{c}}}}$ at point $G_e$. Hence, we now derive equations of the luminaire plane and the lines.

Fig. \ref{xozplanefig} shows the projection of elliptical cone $\Gamma$ on the ${x^{{\rm{a}}}}O{z^{{\rm{a}}}}$ plane. Line segment ${L_1}{L_2}$ is the projection of the complete luminaire circle captured by the camera, and the length of ${L_1}{L_2}$ is equal to the diameter of the luminaire, $2R$. Line segment ${I_1}{I_2}$  is the projection of ${\Pi _e}$ on ${x^{{\rm{a}}}}O{z^{{\rm{a}}}}$ plane, and line segment $L_{1}^{'}L_{2}^{'}$ is the projection of ${\Pi _{{i^*}}}$ on ${x^{{\rm{a}}}}O{z^{{\rm{a}}}}$ plane. Moreover, lines ${{L_1} O^{{\rm{a}}}}$ and ${{L_2}O^{{\rm{a}}}}$ can be expressed as \cite{3Dcirclelocation1992}
\begin{equation}\label{lineOAOB}
{z^{{\rm{a}}}} =  \pm \sqrt { - \frac{{{\lambda _1}}}{{{\lambda _3}}}} {x^{{\rm{a}}}}.
\end{equation}

\begin{figure}[t]
\centering
\includegraphics[height=7cm,trim=0 0 0 0,clip]{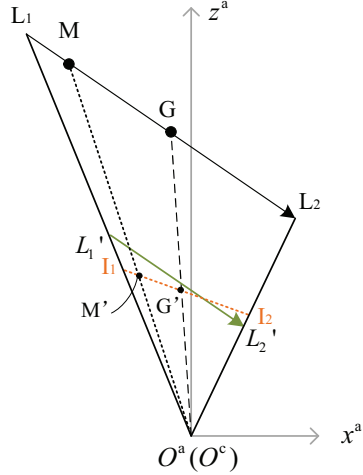}
\caption{The projection of elliptical cone on the ${x^{{\rm{a}}}}O{z^{{\rm{a}}}}$ plane.}\label{xozplanefig}
\end{figure}

By combining and solving ${z^{{\rm{a}}}} = {k_{i^*}}{x^{{\rm{a}}}} + b$ and (\ref{lineOAOB}), the coordinates of points $L_{1}^{'}$ and $L_{2}^{'}$ can be obtained. Here, we define $\boldsymbol{P}_{{L}{{_{1}^{\mathbf{'}}}}}^{\text{a}} = {\left(
{x_{L_{1}^{'}}^{{\rm{a}}}},{y_{L_{1}^{'}}^{{\rm{a}}}},{z_{L_{1}^{'}}^{{\rm{a}}}} \right)^{\rm{T}}}$
and $\boldsymbol{P}_{{L}{{_{2}^{\mathbf{'}}}}}^{\text{a}} = {\left(
{x_{L_{2}^{'}}^{{\rm{a}}}},{y_{L_{2}^{'}}^{{\rm{a}}}},{z_{L_{2}^{'}}^{{\rm{a}}}} \right)^{\rm{T}}}$.
The expression of luminaire plane can be given by
\begin{equation}\label{ABplane}
{z^{{\rm{a}}}} = {k_{{\rm{LED}}}}{x^{{\rm{a}}}} + {b_{{\rm{LED}}}},
\end{equation}
where $k_{{\rm{LED}}}=k_{i^*}$ and ${b_{{\rm{LED}}}} = \frac{{2R \cdot b}}{\left\| \boldsymbol{P}_{{L}{{_{1}^{\mathbf{'}}}}}^{\text{a}} - \boldsymbol{P}_{{L}{{_{2}^{\mathbf{'}}}}}^{\text{a}} \right\|}$.

The coordinate, ${\boldsymbol{P}}_{M'}^{\rm{c}}$, of point $M'$ in CCS can be obtained by (\ref{ICStoCCS}). Then, the coordinate, ${\boldsymbol{P}}_{M'}^{{\rm{a}}} = {\left( {x_{M'}^{{\rm{a}}}},{y_{M'}^{{\rm{a}}}},{z_{M'}^{{\rm{a}}}} \right)^{\rm{T}}}$, of point $M'$ in ACS can be further obtained by ${\boldsymbol{P}}_{M'}^{{\rm{a}}} = {\boldsymbol{R}}_{\rm{c}}^{{\rm{a}}} \cdot {\boldsymbol{P}}_{M'}^{\rm{c}}$. With ${\boldsymbol{P}}_{M'}^{{\rm{a}}} = {\left( {x_{M'}^{{\rm{a}}}},{y_{M'}^{{\rm{a}}}},{z_{M'}^{{\rm{a}}}} \right)^{\rm{T}}}$, the point-normal form equation of line $M'{O^{{\rm{a}}}}$ can be given by
\begin{equation}\label{M'Oline}
\vspace{-0.1cm}
\frac{{{x^{{\rm{a}}}}}}{{x_{M'}^{{\rm{a}}}}} = \frac{{{y^{{\rm{a}}}}}}{{y_{M'}^{{\rm{a}}}}} = \frac{{{z^{{\rm{a}}}}}}{{z_{M'}^{{\rm{a}}}}}.
\vspace{-0.1cm}
\end{equation}
From Fig. \ref{xozplanefig}, we can also observe that luminaire plane and line $M'{O^{{\rm{a}}}}$ intersect at point $M$. Thus by combining and solving (\ref{ABplane}) and (\ref{M'Oline}), we can obtain the coordinate of point $M$ in ACS, i.e., ${\boldsymbol{P}}_{M}^{{\rm{a}}}$. Similarly, we can also obtain the coordinate of point $G_e$, ${\boldsymbol{P}}_{G_e}^{{\rm{a}}}$, in ACS by using the line $G'{O^{{\rm{a}}}}$ and (\ref{ABplane}).
Finally, the coordinates of points $G_e$ and $M$ in CCS, ${\boldsymbol{P}}_{G_e}^{{\rm{c}}}$ and ${\boldsymbol{P}}_{M}^{{\rm{c}}}$, are given by ${\boldsymbol{P}}_{G_e}^{{\rm{c}}} = {\boldsymbol{R}}_{{\rm{a}}}^{\rm{c}} \cdot {\boldsymbol{P}}_{G_e}^{{\rm{a}}}$ and ${\boldsymbol{P}}_{M}^{{\rm{c}}} = {\boldsymbol{R}}_{{\rm{a}}}^{\rm{c}} \cdot {\boldsymbol{P}}_{M}^{{\rm{a}}}$, respectively.

\subsection{Pose and Location Estimation}\label{A-POSELOCATION}\
Based on ${\boldsymbol{n}}_{{\rm{LED}}}^{\rm{c}}$, ${\boldsymbol{P}}_{G_e}^{{\rm{c}}}$ and ${\boldsymbol{P}}_{M}^{{\rm{c}}}$, the pose and location of the user in WCS can be further obtained with VLC as follows.
%

Based on single-view geometry theory, the relationship between ${\boldsymbol{n}}_{{\rm{LED}}}^{\rm{c}} = {( {t_1 },{t_2 },{t_3 })}^{\rm{T}}$ and ${\boldsymbol{n}}_{{\rm{LED}}}^{\rm{w}} = {(0,0,-1)}^{\rm{T}}$ is ${\boldsymbol{n}}_{{\rm{LED}}}^{\rm{c}} = \left({\boldsymbol{R}}_{\rm{c}}^{\rm{w}}\right)^{\rm{T}} \cdot {\boldsymbol{n}}_{{\rm{LED}}}^{\rm{w}}$, which can further derived as 
\begin{equation}\label{Raccordingnv}
\left\{ \begin{array}{l}
\sin \theta = {t_1},   \\
-\cos \theta \sin \varphi = {t_2}, \\
-\cos \theta \cos \varphi = {t_3} .
\end{array} \right.
\end{equation}
The rotation angles $\varphi$ and $\theta$ corresponding to the $x^c$-axis and $y^c$-axis can be obtained by solving (\ref{Raccordingnv}).
In addition, $\overrightarrow {G_eM}$ can be normalized as ${\boldsymbol{n}}_{G_eM}^{\rm{w}} = {\left( {0,1,0} \right)^{\rm{T}}}$ in WCS. The normalized form of $\overrightarrow {G_eM}$ will be ${\boldsymbol{n}}_{G_eM}^{\rm{c}} = \frac{{{{\boldsymbol{P}}_{G_e}^{{\rm{c}}} - {\boldsymbol{P}}_{M}^{{\rm{c}}}}}}{{\left\| {{\boldsymbol{P}}_{G_e}^{{\rm{c}}} - {\boldsymbol{P}}_{M}^{{\rm{c}}}} \right\|}} = {\left( {{s_1},{s_2},{s_3}} \right)^{\rm{T}}}$. We also have ${\boldsymbol{n}}_{G_eM}^{\rm{w}} = {\boldsymbol{R}}_{\rm{c}}^{\rm{w}} \cdot {\boldsymbol{n}}_{G_eM}^{\rm{c}}$, which can be further derived as
\begin{equation}\label{Raccordinghor}
\left\{ \begin{array}{l}
 \sin \psi \cos \theta = {s_1}, \\
 \cos \psi \cos \varphi  + \sin \psi \sin \theta \sin \varphi = {s_2}, \\
 - \cos \psi \sin \varphi  + \sin \psi \sin \theta \cos \varphi = {s_3}.
\end{array} \right.
\end{equation}
Then, with the obtained $\varphi$ and $\theta$, the rotation angle $\psi$ corresponding to $z^c$-axis can be obtained from (\ref{Raccordinghor}). ${\boldsymbol{R}}_{\rm{c}}^{\rm{w}}$ can also be obtained according to (\ref{Rcw}).
By substituting the obtained WCS and CCS coordinates of $G$ and $M$, i.e., ${\boldsymbol{P}}_{G_e}^{{\rm{w}}}$, ${\boldsymbol{P}}_{G_e}^{{\rm{c}}}$, ${\boldsymbol{P}}_{M}^{{\rm{w}}}$ and ${\boldsymbol{P}}_{M}^{{\rm{c}}}$, into
\begin{equation}\label{T}
{\boldsymbol{t}}_{\bf{c}}^{\rm{w}} = \frac{1}{2}\left[ {\left( {{\boldsymbol{P}}_{G_e}^{{\rm{w}}} - {\boldsymbol{R}}_{\rm{c}}^{\rm{w}} \cdot {\boldsymbol{P}}_{G_e}^{{\rm{c}}}} \right) + \left( {{\boldsymbol{P}}_{M}^{{\rm{w}}} - {\boldsymbol{R}}_{\rm{c}}^{\rm{w}} \cdot {\boldsymbol{P}}_{M}^{{\rm{c}}}} \right)} \right],
\end{equation}
${\boldsymbol{t}}_{\rm{c}}^{\rm{w}}$ can be obtained. In (\ref{T}), we average ${\boldsymbol{t}}_{\rm{c}}^{\rm{w}} = {{\boldsymbol{P}}_{G_e}^{{\rm{w}}} - {\boldsymbol{R}}_{\rm{c}}^{\rm{w}} \cdot {\boldsymbol{P}}_{G_e}^{{\rm{c}}}}$ and ${\boldsymbol{t}}_{\rm{c}}^{\rm{w}} = {{\boldsymbol{P}}_{M}^{{\rm{w}}} - {\boldsymbol{R}}_{\rm{c}}^{\rm{w}} \cdot {\boldsymbol{P}}_{M}^{{\rm{c}}}}$ to reduce estimation errors.
In this way, the pose and location of the camera, i.e., ${\boldsymbol{R}}_{\rm{c}}^{\rm{w}}$ and ${\boldsymbol{t}}_{\rm{c}}^{\rm{w}}$, have been obtained.

\section{AO-V-PA Positioning Algorithm}\label{AO-V-PA}
In Section \ref{V-PCA}, an accurate positioning algorithm V-PCA has been proposed. However, in practice, the camera may not always be able to capture complete circular luminaire images due to limited FOV and possible occlusions. In such circumstances, the center of luminaire and the mark point cannot be determined by the space-time coding design or artificial marking, and, thus, V-PCA cannot work.
To circumvent these practical challenges, we propose an anti-occlusion positioning algorithm based on V-PCA, called AO-V-PA, which can achieve positioning using incomplete circular luminaire images. In particular, the normal vector of the luminaire in CCS is first estimated using the same process as V-PCA. Then, we fit the ellipses from arcs, and the center of the ellipse is approximated as the projection of the luminaires's center on the image plane. Finally, the pose and location of the user can be obtained based on the estimated centers and the normal vector ${\boldsymbol{n}}_{{\rm{LED}}}^{\rm{c}}$. Since the principle of ${\boldsymbol{n}}_{{\rm{LED}}}^{\rm{c}}$ estimation in OA-V-PA is similar to the one in V-PCA, we will not reproduce it in this section for brevity. The main differences between OA-V-PA and V-PCA are introduced below.

%

\subsection{Center Coordinates in CCS}
The image plane has only incomplete luminaires that are defined as incomplete ellipses. In such circumstances, the artificial mark point may not be captured, and the space-time coding can also not be completely received. Thus, we propose to approximate the projection of the luminaire's center using the center of the fitted ellipse, so as to calculate the pose of the user. Although the approximation of the projection center may yield additional estimation errors, it can increase the feasibility  and robustness of the positioning algorithm in practice, which will be verified in the simulation and experimental results.

We use $G_e$ and $G_f$ to represent the centers of two circular luminaires, which are captured by the camera and lie on two ellipses, $\Pi_e$ and $\Pi_f$, respectively on the image plane. The projections of $G_e$ and $G_f$ on the image plane are $G'_e$ and $G'_f$, respectively.
The general form functions of ellipses $\Pi_e$ and $\Pi_f$ can be estimated according to (\ref{imgellipse}). Then, the known parameters $A$, $B$, $C$, $D$, and $E$  can be used to obtain the centers of the ellipses.
The coordinate of the center of ellipse $\Pi_e$ in ICS can be expressed as
\begin{equation}\label{ellipsecenter}
\left\{ \begin{array}{l}
x_{e}^{\text{i}}=\frac{BE-2CD}{4AC-{{B}^{2}}}, \\
y_{e}^{\text{i}}=\frac{BD-2AE}{4AC-{{B}^{2}}},
\end{array} \right.
\end{equation}
and it can be used to approximate $\boldsymbol{P}_{G'_e}^{\text{i}} \approx {{\left( x_{e}^{\text{i}},y_{e}^{\text{i}} \right)}^{\text{T}}}$. Similarly, the center of another fitted ellipse $\Pi_f$ on the image plane can also be approximated with $\boldsymbol{P}_{G'_f}^{\text{i}} \approx {{\left( x_{f}^{\text{i}},y_{f}^{\text{i}} \right)}^{\text{T}}}$. Next, $\boldsymbol{P}_{G'_e}^{\text{i}}$ and $\boldsymbol{P}_{G'_f}^{\text{i}}$ are respectively used to obtain the point-normal form equations of lines $G'_eO^{\rm{a}}$ and $G'_fO^{\rm{a}}$, which are combined and solved with (\ref{ABplane}) to obtain the coordinates of the point $G_e$ and $G_f$ in ACS . Thus, we obtain the coordinates,  ${{\boldsymbol{P}}_{G_e}^{\rm{a}}}$ and ${{\boldsymbol{P}}_{G_f}^{\rm{a}}}$, which are respectively substituted into ${{\boldsymbol{P}}_{G_e}^{\rm{c}}} = {\boldsymbol{R}}_{{\rm{a}}}^{\rm{c}} \cdot {{\boldsymbol{P}}_{G_e}^{{\rm{a}}}}$ and ${{\boldsymbol{P}}_{G_f}^{\rm{c}}} = {\boldsymbol{R}}_{{\rm{a}}}^{\rm{c}} \cdot {{\boldsymbol{P}}_{G_f}^{{\rm{a}}}}$ to obtain ${{\boldsymbol{P}}_{G_e}^{\rm{c}}}$ and ${{\boldsymbol{P}}_{G_f}^{\rm{c}}}$. Consequently, the coordinates of two luminaires' centers in CCS can be estimated.

\subsection{Pose and Location Estimation}
This step estimates the pose and location of the user. Similar to Section \ref{V-PCA}, we first calculate the rotation matrix. The rotation angles $\varphi$ and $\theta$ corresponding to the $x^c$-axis and $y^c$-axis can be obtained according to (\ref{Raccordingnv}). The rotation angle $\psi$ corresponding to the $z^{\rm{c}}$-axis can be obtained by using the estimated luminaires' centers, i.e., ${{G}_{e}}$ and ${{G}_{f}}$.
In particular, the normalized form of $\overrightarrow{{{G}_{e}}{{G}_{f}}}$ in WCS is calculated first with the known ${{\boldsymbol{P}}_{G_e}^{\rm{w}}}$ and ${{\boldsymbol{P}}_{G_f}^{\rm{w}}}$, and we denote it by $\boldsymbol{n}_{{{G}_{e}}{{G}_{f}}}^{\text{w}}={{\left( {{g}_{1}},{{g}_{2}},{{g}_{3}} \right)}^{\text{T}}}$ for simplicity.
Note that ${{\left( {{g}_{1}},{{g}_{2}},{{g}_{3}} \right)}^{\text{T}}}$ cannot be equal to ${\left( {0,1,0} \right)^{\rm{T}}}$. This is because that $\overrightarrow{{{G}_{e}}{{G}_{f}}}$ can have arbitrary direction since the camera may capture two luminaires at any locations. Then, the normalized form of $\overrightarrow{{{G}_{e}}{{G}_{f}}}$ in CCS is calculated by  ${{\boldsymbol{n}}_{{{G}_{e}}{{G}_{f}}}^{\rm{c}}} = \frac{{{{\boldsymbol{P}}_{G_e}^{\rm{c}}} - {{\boldsymbol{P}}_{G_f}^{\rm{c}}}}}{{\left\| {{{\boldsymbol{P}}_{G_e}^{\rm{c}}} - {{\bf{P}}_{G_f}^{\rm{c}}}} \right\|}} = {\left( {{h_1},{h_2},{h_3}} \right)^{\rm{T}}}$. According to ${\boldsymbol{n}}_{{{G}_{e}}{{G}_{f}}}^{\rm{c}} = \left({\boldsymbol{R}}_{\rm{c}}^{\rm{w}}\right)^{\rm{T}} \cdot {\boldsymbol{n}}_{{{G}_{e}}{{G}_{f}}}^{\rm{w}}$, we have
\begin{equation}\label{Raccortinghorcomplex}
\left\{ \begin{array}{l}
{g_1}\cos \psi \cos \theta  + {g_2}\sin \psi \cos \theta  - {g_3}\sin \theta = {h_1} , \\
- {g_1}\left( {\sin \psi \cos \varphi  + \cos \psi \sin \theta \sin \psi } \right) + {g_2}\left( {\cos \psi \cos \varphi  + \sin \psi \sin \theta \sin \psi } \right) + {g_3}\cos \theta \sin \varphi = {h_2},  \\
 {g_1}\left( {\sin \psi \sin \varphi  + \sin \psi \sin \theta \cos \varphi } \right) - {g_2}\left( {\cos \psi \sin \varphi  + \sin \psi \sin \theta \cos \varphi } \right) + {g_3}\cos \theta \cos \varphi = {h_3} .
\end{array} \right.
\end{equation}
Then, with the obtained $\varphi$ and $\theta$, the rotation angle $\psi$ corresponding to $z^c$-axis can be calculated by using the auxiliary angle formula in trigonometric function to solve (\ref{Raccortinghorcomplex}). ${\boldsymbol{R}}_{\rm{c}}^{\rm{w}}$ can also be obtained according to (\ref{Rcw}).

Similar to (\ref{T}), ${\boldsymbol{t}}_{\bf{c}}^{\bf{w}}$ can be calculated by substituting the obtained WCS and CCS coordinates of $G_e$ and $G_f$, i.e., ${{\boldsymbol{P}}_{G_e}^{\rm{w}}}$, ${{\boldsymbol{P}}_{G_e}^{\rm{c}}}$, ${{\bf{P}}_{G_f}^{\rm{w}}}$, and ${{\boldsymbol{P}}_{G_f}^{\rm{c}}}$, into
\begin{equation}\label{T2}
{\boldsymbol{t}}_{\bf{c}}^{\bf{w}} = \frac{1}{2}\left[ {\left( {\boldsymbol{P}_{{{G}_{e}}}^{\text{w}} - {\boldsymbol{R}}_{\rm{c}}^{\rm{w}} \cdot \boldsymbol{P}_{{{G}_{e}}}^{\text{c}}} \right) + \left( {\boldsymbol{P}_{{{G}_{f}}}^{\text{w}} - {\boldsymbol{R}}_{\rm{c}}^{\rm{w}} \cdot \boldsymbol{P}_{{{G}_{f}}}^{\text{c}}} \right)} \right].
\end{equation}
Consequently, the pose and location of the user, i.e., ${\boldsymbol{R}}_{\rm{c}}^{\rm{w}}$ and ${\boldsymbol{t}}_{\bf{c}}^{\bf{w}}$, have been estimated.

Above all, we can observe that V-PCA and OA-V-PA share most of the positioning process, while have some key differences. In particular, when a complete circular luminaire and an incomplete circular luminaire are captured by the camera, the V-PCA can be used to estimate the location; if there are only incomplete circular luminaires captured, the OA-V-PA can be used. Here, we term this dynamic integration of V-PCA and OA-V-PA as V-PA, which can select appropriate positioning algorithms according to the captured luminaires.  According to (\ref{Raccortinghorcomplex}), when two luminaires captured by the camera are not coplanar but parallel, V-PA will still be applicable.

\section{Implementation of V-PA}
In this section, we establish an experimental prototype to verify the feasibility and efficiency of the proposed algorithm. The illustration of the experimental setup is shown in Fig. \ref{testbed}, and the device specifications are given in Table \ref{device}.

\begin{figure*}[b]
\centering
\subfigure[Experimental environment.]{\includegraphics[height=7.2cm,trim=0 0 0 0,clip]{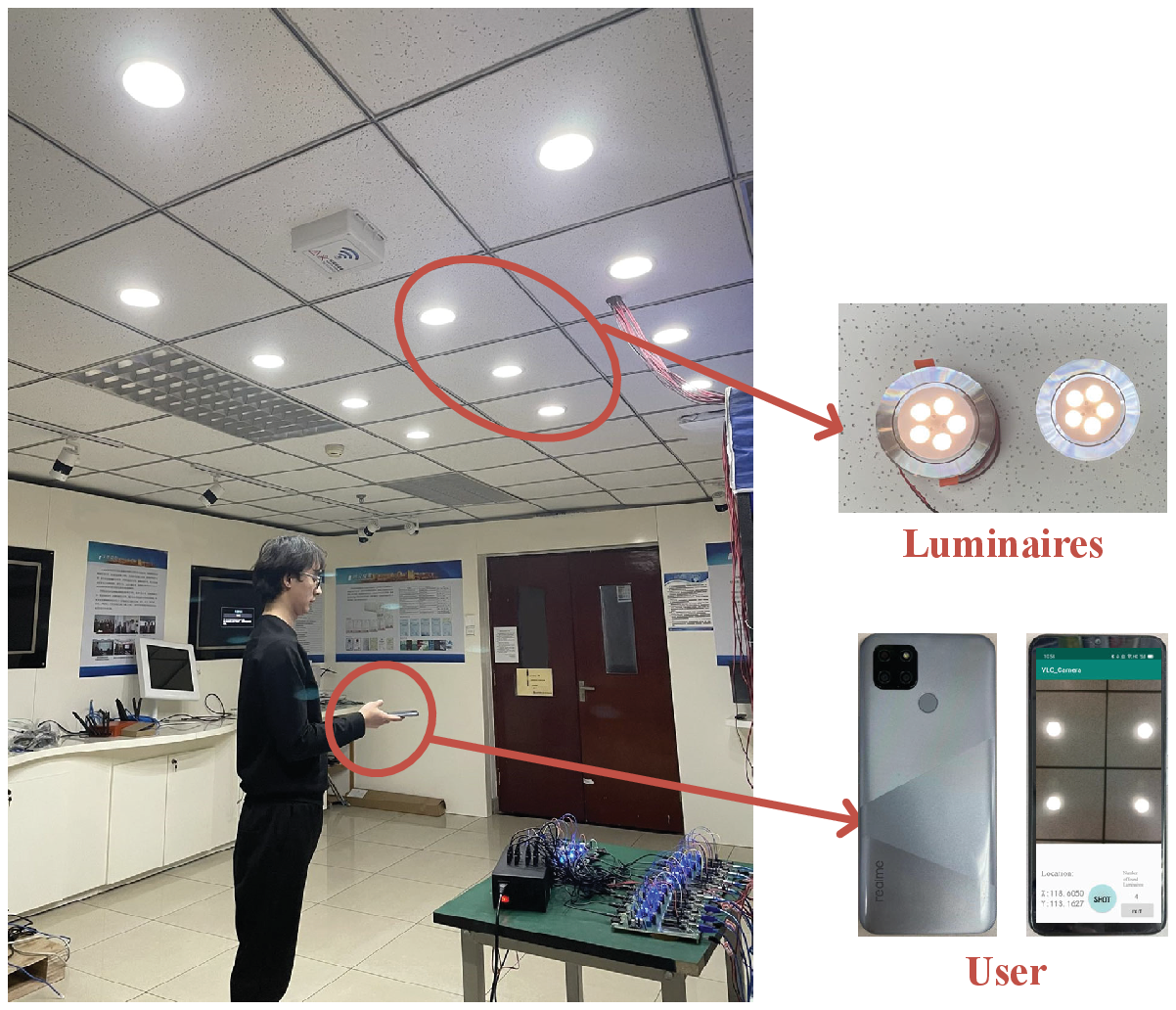}}
\subfigure[Block diagram.]{\includegraphics[height=5.8cm,trim=0 0 0 0,clip]{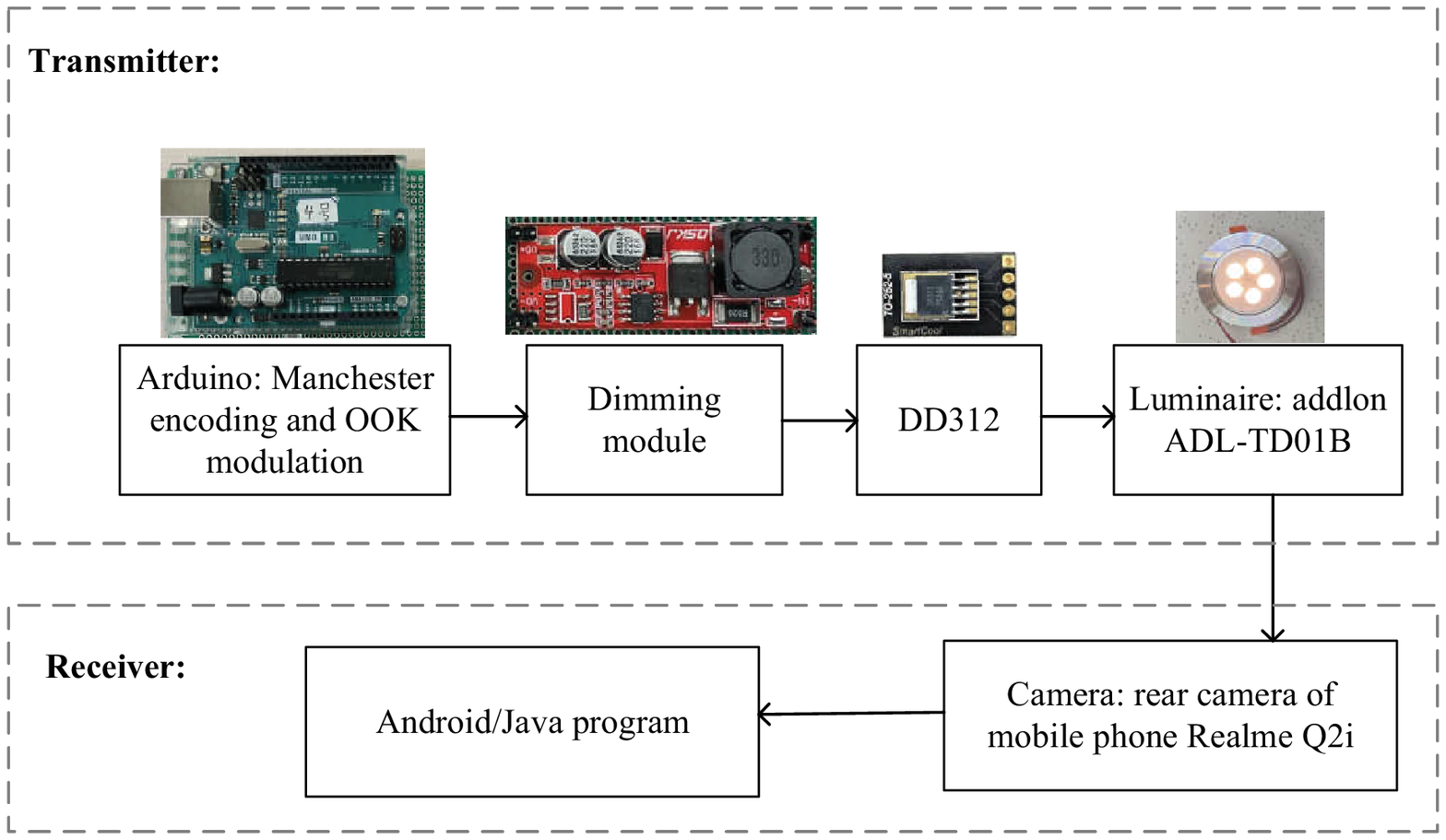}}
\caption{Experimental setup. }\label{testbed}
\end{figure*}
\subsection{Implementations}
The experiment is carried out in an indoor environment. The height of the transmitters is 2.84 m. There are 12 circular luminaires mounted on the ceiling of the room to provide illumination and positioning information. As long as two luminaires are captured, whether they are complete or not, the user can be located. In the experiment, we choose to artificially mark the center and the mark point on the luminaire. All the LEDs in a luminaire transmit the same information containing the luminaire's coordinate and ID.

\begin{table}
		\centering \caption{Device Specifications}\label{device}
            \footnotesize
		\begin{tabular}{|c|c|c|}\hline
			\multicolumn{2}{|c|} {\bf{Parameters}}& {\bf{Model/Values}}\\
            \hline
            \hline
            \multirow{4}{*}{Transmitter}
            & Model & Addlon ADL-TD01B \\
            \cline{2-3}
            & Semi-angle, $\Phi_{1/2}$ & $67.5^\circ$ \\
            \cline{2-3}
            &Transmitter Power, $P_t$ & $5$ $\rm{W}$ \\
            \cline{2-3}
            &Radius  & $3.5$ $\rm{cm}$ \\
            \hline
            \multirow{6}{*}{Camera}
            & Model & Rear camera of Realme Q2i \\
            \cline{2-3}
            &\multirow{3}{*}{Intrinsic parameters}
            & $f=3.462 \rm{mm}$\\
            \cline{3-3}
            & \multirow{3}{*}{} & $d_x=d_y=1.12\times10^{-4}cm$\\
            \cline{3-3}
            & \multirow{3}{*}{} &  $\left(u_0,v_0\right)=\left(2080,1560\right)$ \\
            \cline{2-3}
            & Resolution & $4160 \times 3120$\\
            \cline{2-3}
            & Exposure time & 1.25 \rm{ms} \& 6.67 \rm{ms}\\
            \hline
		\end{tabular}
\end{table}

To successfully transmit the information, Arduino Studio is used to generate signals with Manchester encoding and on-off keying (OOK) modulation. Then, the dimming module is designed to control the transmitted power of the LED, after which, DD312, a constant current LED driver is used to drive the LED to emit light without flicking. In this way, the information of the transmitters are broadcasted to the air.
To successfully receive the geometric information, a single camera is calibrated by a conventional method \cite{zhou2019robust} to extract information from 2D images. The camera we use is the rear Complementary Metal Oxide Semiconductor (CMOS) camera of Realme Q2i. Once an image is captured, image processing is implemented by an Android/Java program, and the region of interest (ROI) is exploited for image processing to enhance the accuracy. Then, after contour extraction and curve fitting, the elliptic curve equation and the pixel coordinates of the projections can be obtained from the processed images.
In addition, to successfully receive the modulated VLC signals that contain the IDs and world coordinates of the luminaires, we leverage the rolling shutter effect of the CMOS image sensor to capture the fringe image of the luminaire, and we demodulate the fringes to receive the VLC information. Finally, the geometric feature information and the VLC information will be input into an Android/Java program, and the pose and location of the user can be output by implementing the proposed algorithm in the program.
\begin{figure*}[t]
\centering
\subfigure[Original image.]{\includegraphics[height=4cm,trim=0 0 0 0,clip]{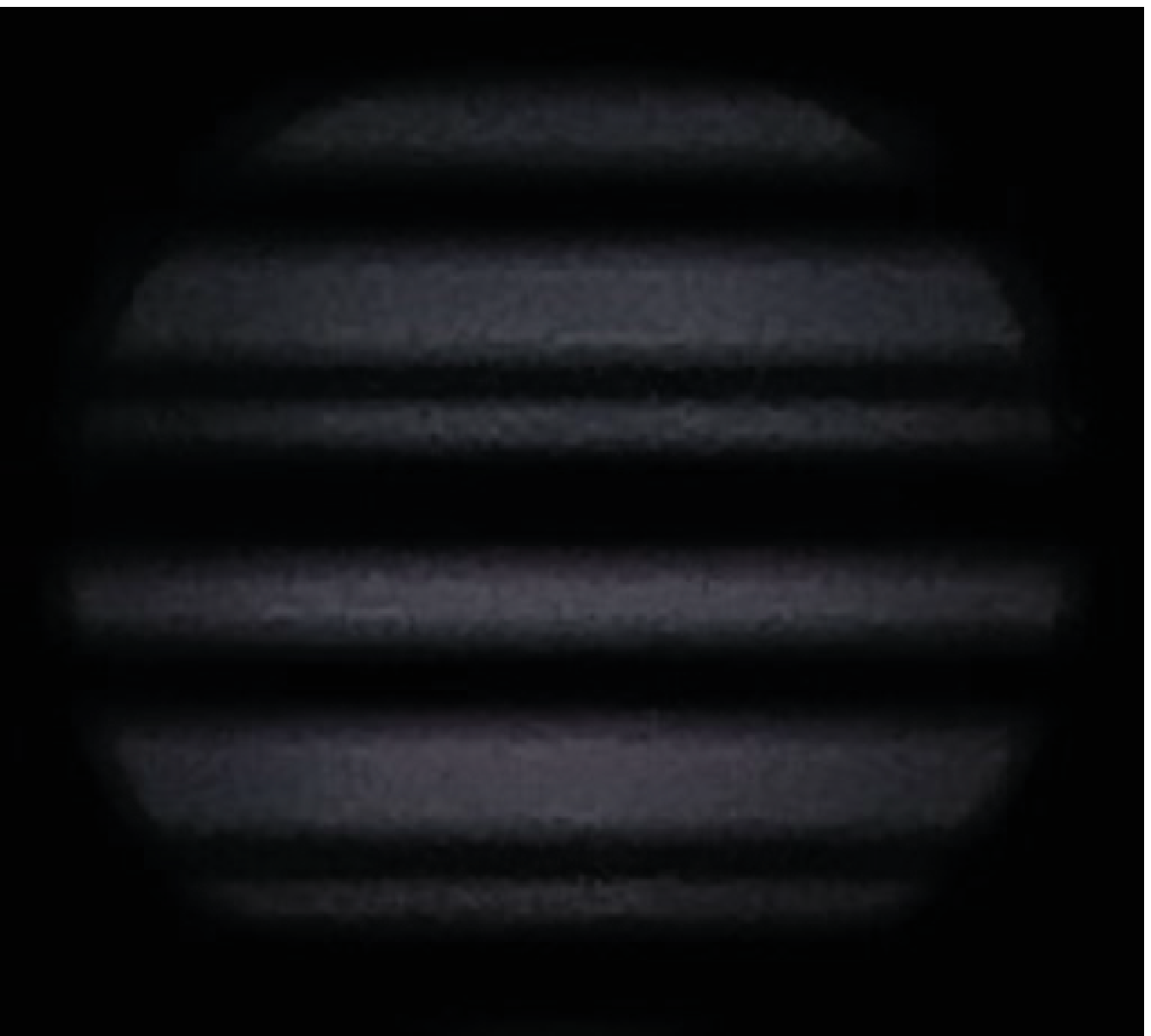}}
\subfigure[Filtered image.]{\includegraphics[height=4cm,trim=0 0 0 0,clip]{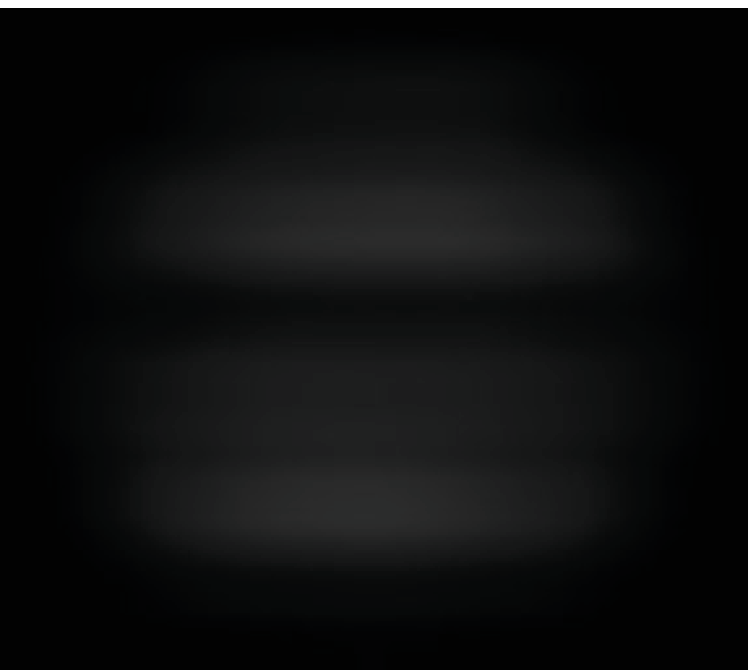}}
\subfigure[Fitted outline.]{\includegraphics[height=4cm,trim=0 0 0 0,clip]{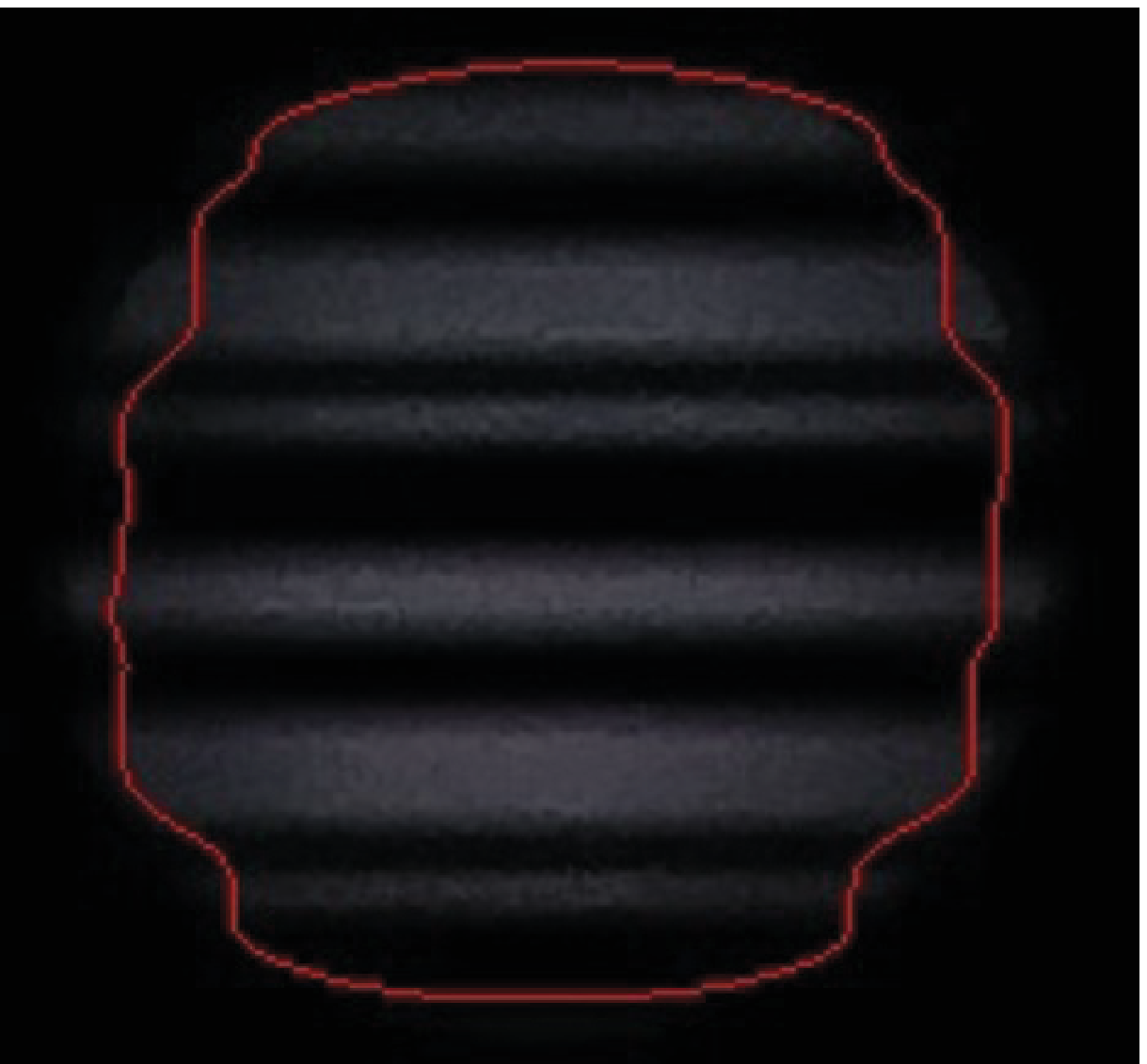}}
\caption{Fringe images in image processing. }\label{fringe}
\end{figure*}

\subsection{Image Processing}
At the receiver side, both geometric feature information and VLC signals are supposed to be obtained. However, the fringe image that is easy to demodulate does not have enough clear outlines to fit the elliptic curve equation.
For instance, Fig. \ref{fringe} (a) is an original fringe image, which will turn into Fig. \ref{fringe} (b) after filtering. It can be observed that the outline of Fig. \ref{fringe} (b) is rather vague, which can lead to further detection errors as shown in Fig. \ref{fringe} (c).
Therefore, we propose a fused image processing scheme, and the flow diagram is shown in Fig. \ref{flow} (a). In the proposed scheme, once a user needs to locate, the image sensor will take two photos under different exposure times (ET), and we develop an application as shown in Fig. \ref{flow} (b) for convenience. In this application, it can be automatically shoot under two ETs, and the interval between two ETs can be artificially controlled.
When this interval is extremely short, the two captured images have approximately the same pose and location, and, thus, the pixel coordinates in two images are the same. Under a long ET, the luminaire is captured as a complete circle or an ellipse, while under a short ET, it is captured as a fringe image, as shown in Fig. \ref{ET} (a) and Fig. \ref{ET} (b), respectively.

\begin{figure*}
  \centering
  \subfigure[Flow diagram.]{\includegraphics[height=10cm]{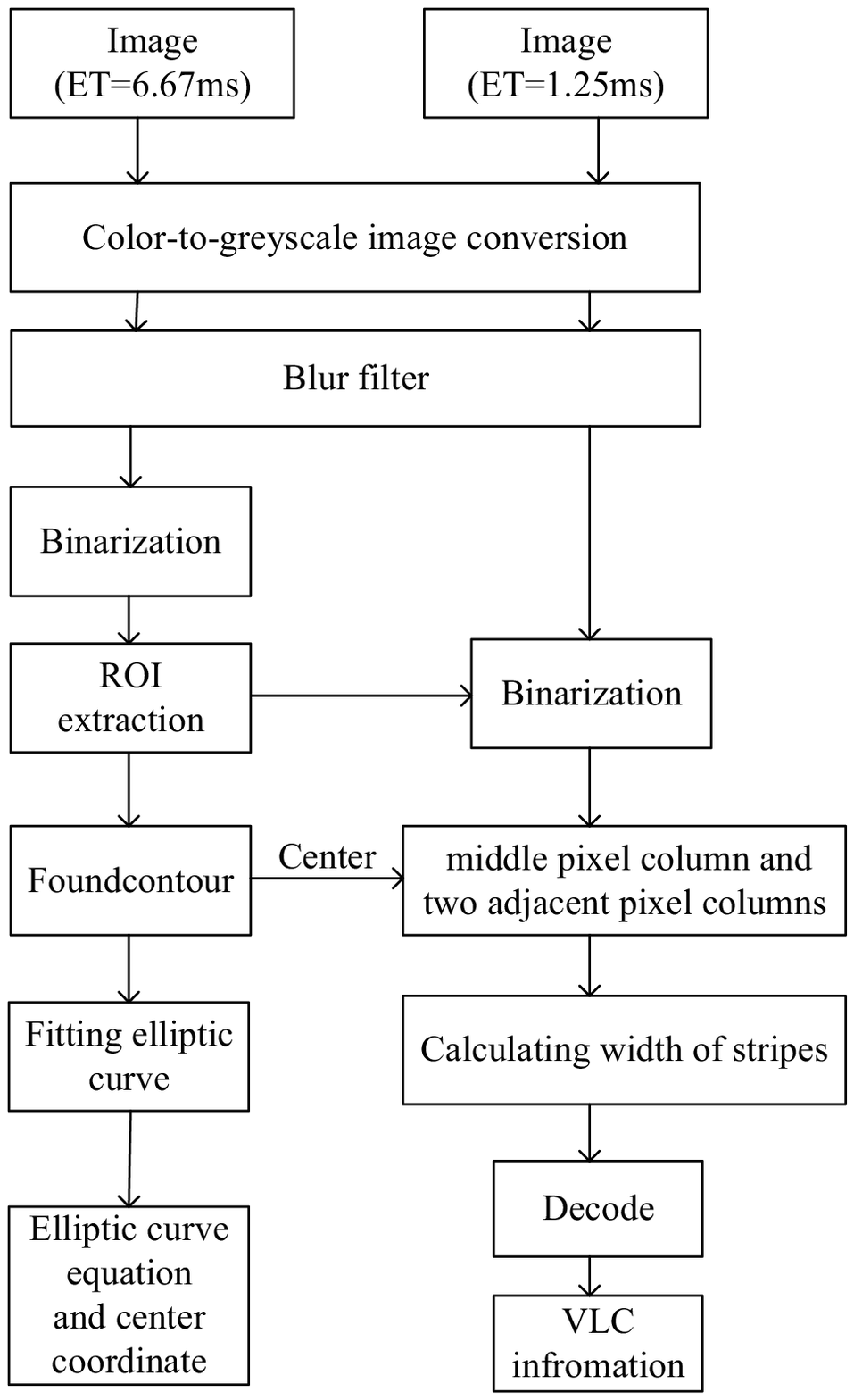}}
  \subfigure[APP interface.]{\includegraphics[height=9.6cm]{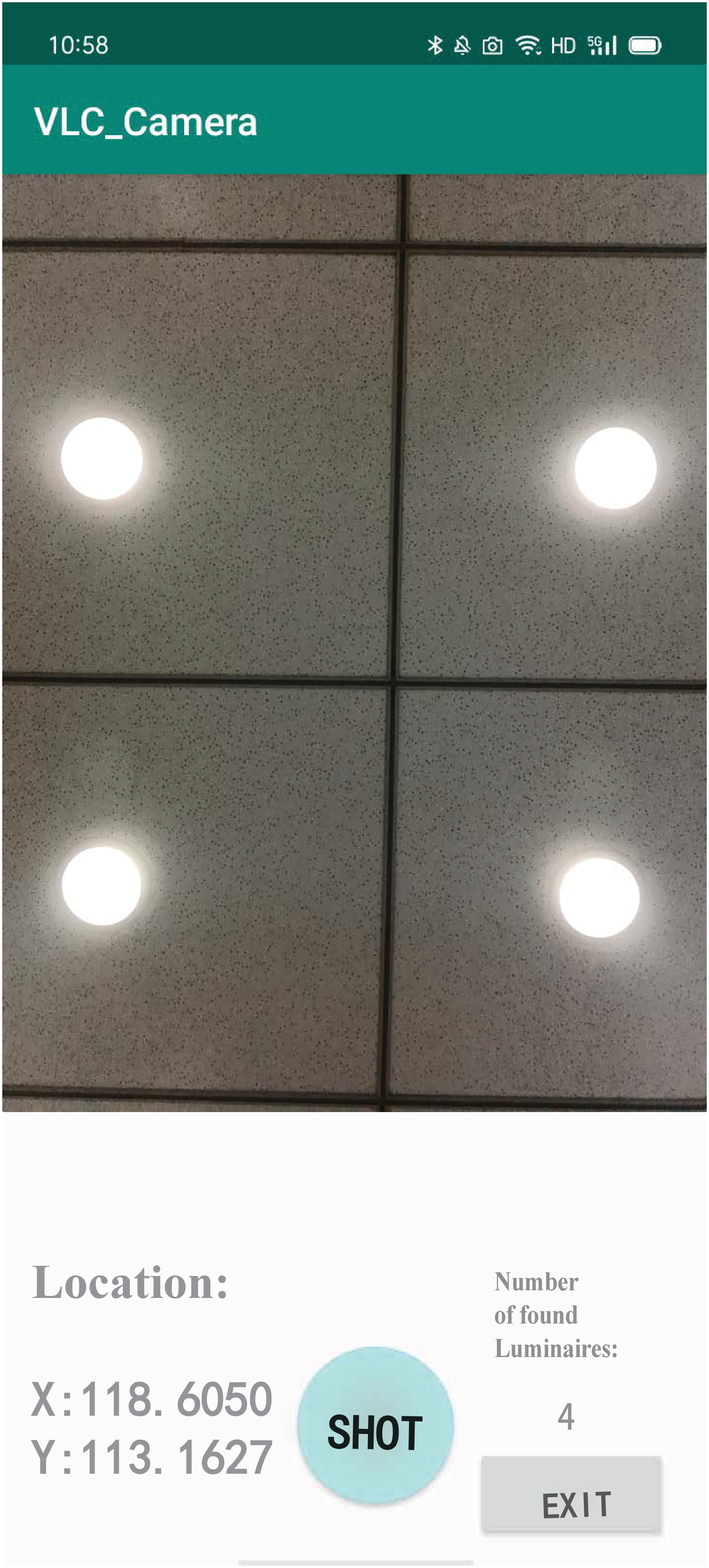}}
  \caption{Flow diagram of image processing and APP interface.}\label{flow}
\end{figure*}

\begin{figure}
\centering
\subfigure[Image captured under a long exposure time (ET=6.67 ms).]{\includegraphics[height=4cm,trim=0 0 0 0,clip]{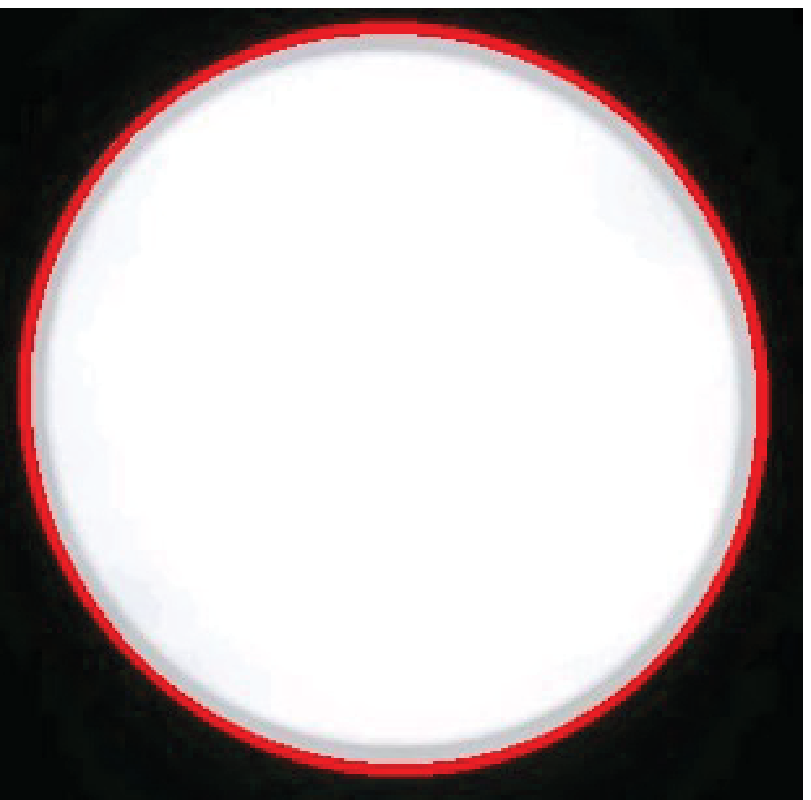}}
\subfigure[Image captured under a short exposure time (ET=1.25 ms).]{\includegraphics[height=4cm,trim=0 0 0 0,clip]{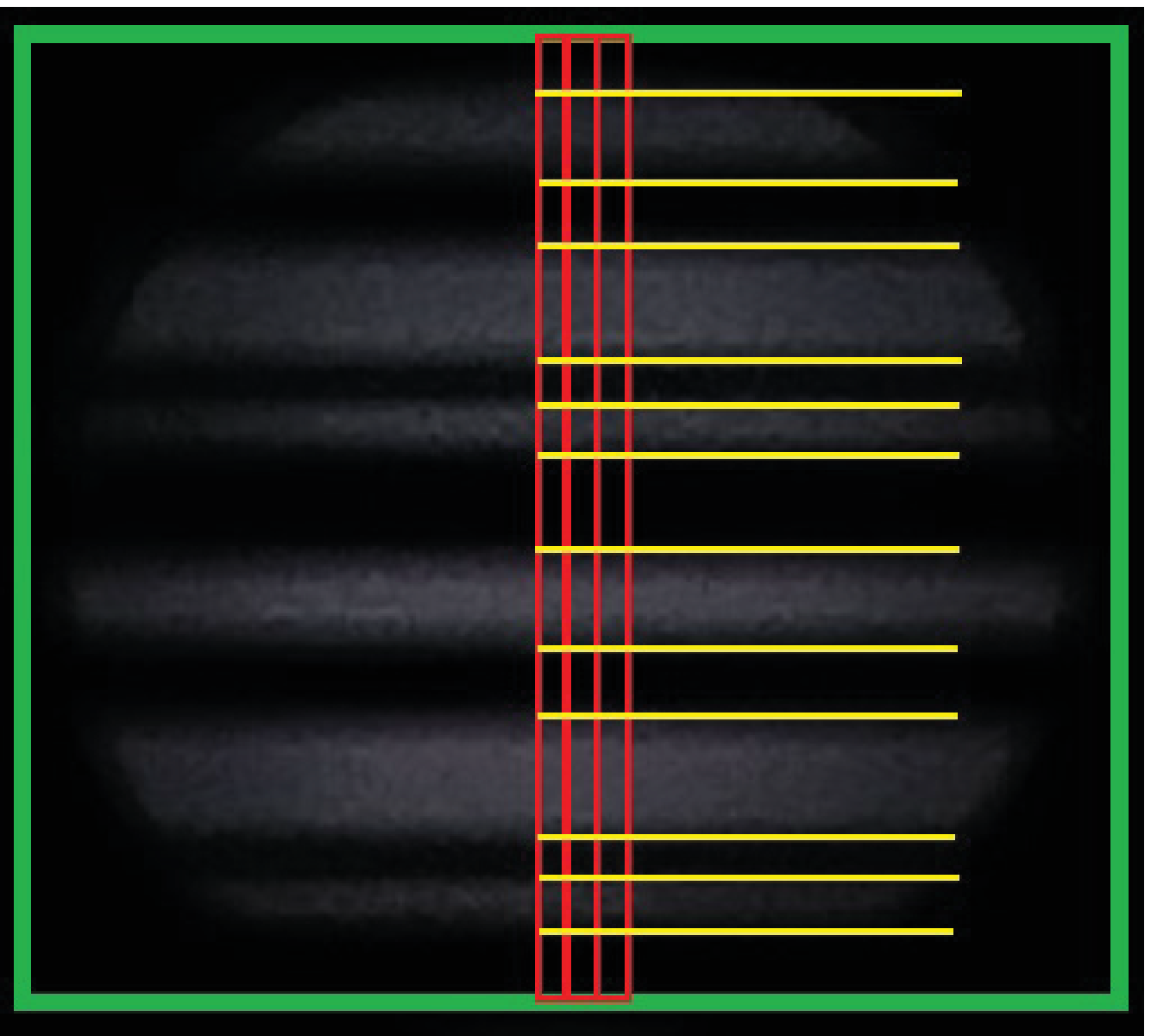}}
\caption{Images in two-step image processing. }\label{ET}
\end{figure}

Fig. \ref{ET} (a) and Fig. \ref{ET} (b) are used to obtain the geometric feature information and VLC signals as illustrated in Fig. \ref{flow} (a). First, Fig. \ref{ET} (a) and Fig. \ref{ET} (b) are converted from RGB images to greyscale images. For the images under the long ET, the filtered image is binarized to make the area of the bright region clear, so as to facilitate the ROI extraction. Since this image is considered to have the same pixels coordinates with the one under the short ET, the extracted ROIs can also be used in the image captured under the short ET.
Then, Hough Transform \cite{ballard1981hough} is used for object detection. The points on the ellipse contour, the projection point of the luminaire's center, and the projection point of the mark point are all extracted, and the extracted points on the contour are used to fit the elliptic curve equation with the least squares fitting. With this equation, the center of the ellipse can also be obtained according to (\ref{ellipsecenter}) when the center or the mark point is not detected.
For images under the short ET, a binarization process is performed in each ROI. The threshold of this binarization is higher than that under the long ET, since bright and dark fringes here have to be more distinguishable to reduce error in decoding. Then, according to the pixel coordinate of the center in the image, the middle pixel column and its two adjacent pixel columns are extracted, as shown in the regions of red boxes in Fig. \ref{ET} (b). The three pixel columns are selected here to reduce the calculation error. The widths of the bright and dark fringes in red boxes are further calculated. Then, the VLC signals are obtained by decoding the width values of the fringes.

\section{Simulation and Experimental Results}
In this section, we evaluate the performance of V-PCA and OA-V-PA via simulation and experimental results.
\subsection{Simulation Setup}
We consider a rectangular room, in which four luminaires are deployed on the ceiling of the room. The system parameters are listed in Table \ref{parameters}. Unless otherwise specified, the radius of the LED luminaire is 15 cm, and the image noise is modeled as a white Gaussian noise having an expectation of zero and a standard deviation of 2 pixels \cite{bai2021computer}. All statistical results are averaged over independent runs of 10,000 samples. For each simulation sample, the location and the pose of the user are generated randomly in the room, and the tilted angle of the user is also generated randomly on the premise that at least two incomplete luminaires can be captured. To mitigate the impact of image noise, the pixel coordinate is obtained by processing 20 images for each location \cite{bai2021computer}.
\begin{table}[htbp]
		\centering \caption{System parameters}\label{parameters}
            \footnotesize
		\begin{tabular}{|c|c|}\hline
			\multicolumn{1}{|c|} {\bf{Parameters}}&\bf{Values}\\\hline
             \hline
			LED semi-angle, ${\Phi _{1/2}}$&$60^\circ$ \\\hline
            Principal point of camera & $\left( {{u_{0,}}{v_0}} \right) = (320,240)$ \\\hline
            Physical sizes  & ${d_x} = {d_y} = 1.25 \times 10^{-3}$ \\\hline
            Room size & 8 ${\rm{m}}$  $\times$ 6 $ {\rm{m}}$ $\times$ 3 $ {\rm{m}}$ \\\hline
            Location of two LEDs $(\rm{m})$ & (2,2,3), (6,2,3),(2,4,3), (6,4,3) \\\hline
		\end{tabular}
\end{table}

We evaluate the positioning performance in terms of location and pose accuracy. We define the location error as $E_\textrm{loc} = \left\| {{\boldsymbol{r}}_{{\rm{true}}}^{\rm{w}} - {\boldsymbol{r}}_{{\rm{est}}}^{\rm{w}}} \right\|$, where ${\boldsymbol{r}}_{{\rm{true}}}^{\rm{w}} = \left( {x_{{\rm{true}}}^{\rm{w}},y_{{\rm{true}}}^{\rm{w}},z_{{\rm{true}}}^{\rm{w}}} \right)^{\rm{T}}$ and ${\boldsymbol{r}}_{{\rm{est}}}^{\rm{w}} = \left( {x_{{\rm{est}}}^{\rm{w}},y_{{\rm{est}}}^{\rm{w}},z_{{\rm{est}}}^{\rm{w}}} \right)^{\rm{T}}$ are true and estimated world coordinates of the user, respectively.
In addition, with the true rotation ${\boldsymbol{R}}_{\rm{c,true}}^{\rm{w}}$, we quantify the relative error of the estimated pose, ${\boldsymbol{R}}_{\rm{c,est}}^{\rm{w}}$, by ${E_\textrm{pos}}\left( \%  \right) = \left\| {{{\boldsymbol{q}}_{{\rm{true}}}} - {\bf{q}}_{\rm{est}}} \right\|/\left\| {\boldsymbol{q}}_{\rm{est}} \right\|$ \cite{EPnE2009}, where ${{\boldsymbol{q}}_{{\rm{true}}}}$ and ${\boldsymbol{q}}_{\rm{est}}$ are the normalized quaternions of the true and the estimated rotation matrices, i.e., ${\boldsymbol{R}}_{\rm{c,true}}^{\rm{w}}$ and ${\boldsymbol{R}}_{\rm{c,est}}^{\rm{w}}$, respectively.

\subsection{Simulation Results}
\subsubsection{Effect of image noise and the luminaire's radius on positioning performance}
In practice, since two captured luminaires can be complete or incomplete, both V-PCA and OA-V-PA algorithms may be used. Therefore, we evaluate the performance of V-PA that adaptively uses V-PCA and OA-V-PA according to the practical scenario, and we compare the positioning accuracy of V-PA with other existing algorithms. In particular, we conduct the IMU-based VLP (V-IMU) algorithm \cite{IMU2019} and PnP algorithm as baseline schemes. V-IMU also uses the circular luminaire for positioning, and the PnP is a typical positioning algorithm in computer vision field. Note that IMU typically measures the pitch and roll angles with an error of ${{{1.5}^ \circ }}$, and the azimuth angle with an error of ${{{15}^ \circ }}$ \cite{partIMU2019PJ}. Therefore, for the V-IMU algorithm, we impose random measurement errors that satisfy the uniform distribution of $\left[ {0,{{1.5}^ \circ }} \right]$ on the pitch and roll angles, and random measurement errors that satisfy the uniform distribution of $\left[ {0,{{15}^ \circ }} \right]$  on the azimuth angle. In addition, the PnP algorithm requires at least four LEDs for positioning. However, not all samples capture four luminaires simultaneously. Thus we select four LEDs evenly from two arcs captured by the camera, such as the LEDs at the locations of $M$, $N$, $P$ and $Q$ on the luminaire in Fig. \ref{spacetime}.

Figure \ref{CDF-3types} compares the performance of our proposed algorithm with PnP and V-IMU algorithms in terms of the cumulative distribution function (CDF) of the location error. From this figure, we can observe that V-PA has the best performance among the three algorithms. As shown in Fig. \ref{CDF-3types}, V-PA is able to achieve a 90th percentile accuracy of about 10 cm. In contrast, the PnP algorithm can achieve a 78th percentile accuracy of about 10 cm. Meanwhile, for the V-IMU algorithm, only a 16th percentile accuracy of about 10 cm is achieved. This is because the V-IMU algorithm requires the tilted angles to be small to achieve reliable approximation for the projection radius.

\begin{figure}[t]
\centering
\includegraphics[height=6.6cm,trim=0 0 0 0,clip]{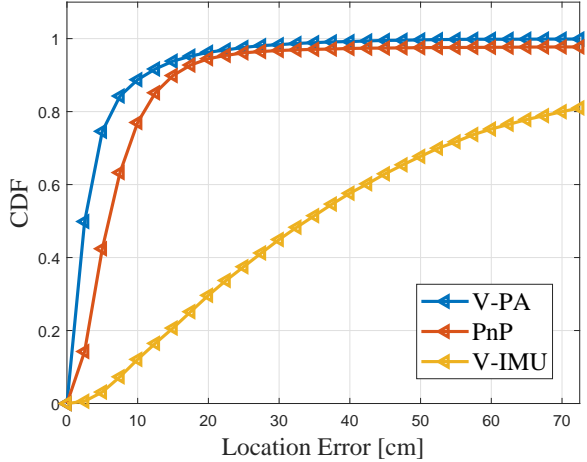}
\caption{CDF of the location error.}\label{CDF-3types}
\end{figure}
\begin{figure}[t]
\centering
\subfigure[]{\includegraphics[height=6cm,trim=0 0 6 10,clip]{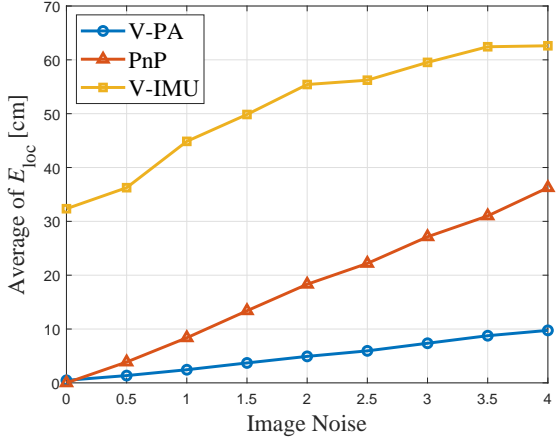}}
\subfigure[]{\includegraphics[height=6cm,trim=0 0 6 10,clip]{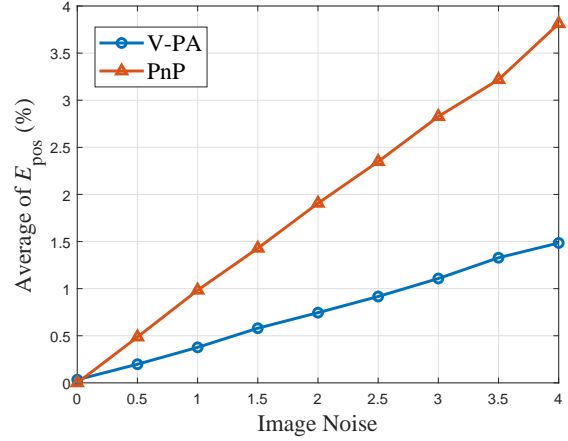}}
\caption{The mean of $E_{\rm{loc}}$ and $E_{\rm{pos}}$ versus the image noise. }\label{imagenoise}
\end{figure}

Figure \ref{imagenoise} shows the effect of image noise on V-PA in terms of the average $E_\textrm{loc}$ and $E_\textrm{pos}$. The image noise is modeled as a white Gaussian noise with an expectation of zero and a standard deviation ${\sigma _n}$ that ranges from 0 to 4 pixels. From Fig. \ref{imagenoise} (a), we can observe that the average $E_\textrm{loc}$ is 0 cm when ${\sigma _n}$ is 0 for V-PA and PnP. That indicates that the location errors of them are totally caused by the image noise. The average $E_\textrm{loc}$ of V-IMU increases from 31 cm to 61 cm as the image noise increases from 0 to 4 pixels, while that of PnP increases from 0 cm to 37 cm. Besides, as the image noise increases, the average $E_\textrm{loc}$ of V-PA varies in a relative small range, from 0 cm to 10 cm, which indicates that V-PA is more robust to image noise. In addition, Fig. \ref{imagenoise} (b) further illustrates the average $E_\textrm{pos}$ versus the image noise. The average $E_\textrm{pos}$ of V-PA increases from $0\%$ to $1.5\%$ as radius increases, while that of PnP increases from $0\%$ to $3.8\%$. These results also verify that, compared with the PnP algorithm, V-PA is more robust to image noise.

\begin{figure}[tbp]
\centering
\subfigure[]{\includegraphics[height=6cm,trim=5 0 10 10,clip]{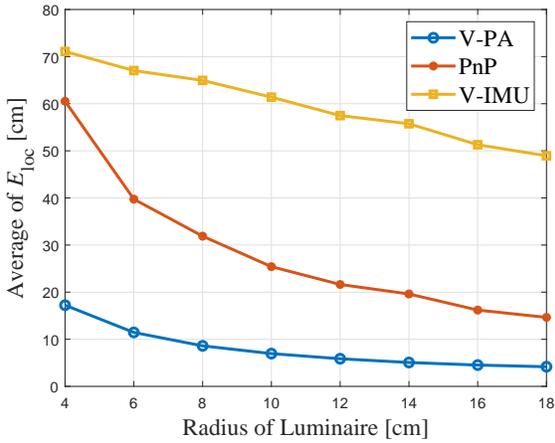}}
\subfigure[]{\includegraphics[height=6cm,trim=6 0 10 10,clip]{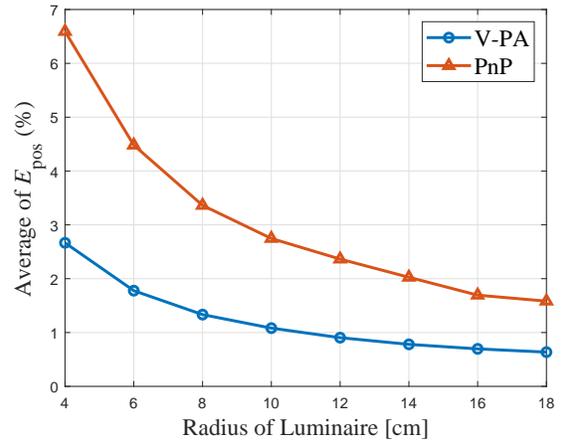}}
\caption{The mean of $E_{\rm{loc}}$ and $E_{\rm{pos}}$ versus the radius of the luminaire. }\label{radius}
\end{figure}

Figure \ref{radius} evaluates the effect of the radius of the luminaire on positioning accuracy. This performance is captured by the average $E_\textrm{loc}$ and the average $E_\textrm{pos}$ with the radius varying from 4 cm to 18 cm. As shown in Fig. \ref{radius} (a), the accuracy of location estimation improves as the radius of luminaire increases. V-PA has the best performance among the three algorithms. For V-PA, the average $E_\textrm{loc}$ remains below 18 cm for all radii. For PnP, the average $E_\textrm{loc}$ decreases from 61 cm to 15 cm as the radius of the luminaire increases. Meanwhile, for V-IMU, the average $E_\textrm{loc}$ decreases from 70 cm to 49 cm. We can observe that, when the radius is below 8 cm, the average $E_\textrm{loc}$ of both V-PA and PnP decreases fast as radius increases. This is due to the fact that, when the captured geometric features have small size, the positioning accuracy is affected by the image noise more seriously. Fig. \ref{radius} (b) compares the average $E_\textrm{pos}$ of V-PA and PnP. Here, we observe V-PA is more accurate than PnP. The average $E_\textrm{pos}$ of V-PA decreases from $2.8\%$ to $0.7\%$ as the radius increases, while that of PnP decreases from $6.6\%$ to $1.7\%$. This also indicates that the pose accuracy of the PnP algorithm is affected more by the small size of luminaire than that of the V-PA algorithm.


\subsubsection{Effect of the arc length of the captured luminaires on positioning performance}
We then also compare the positioning performance of V-PCA and OA-V-PCA algorithms, in which, we evaluate how the length of the contours that extracted from the captured luminaires affect the positioning performance. For comparison, we conduct four schemes: i) OA-V-PA with two semicircles extracted, ii) OA-V-PA with two superior arcs extracted, iii) V-PCA with a circle and a semicircle extracted, and iv) V-PCA with two circles extracted.

\begin{figure}[t]
\centering
\includegraphics[height=6.6cm,trim=0 0 0 0,clip]{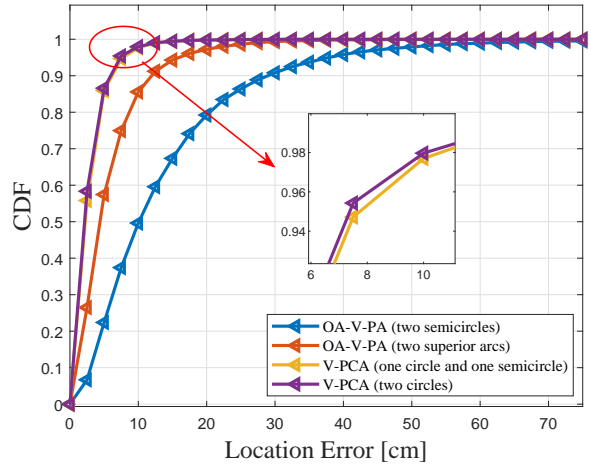}
\caption{Cumulative Distribution Function (CDF) of the location error.}\label{CDF-diffarc}
\end{figure}

Figure \ref{CDF-diffarc} compares the performance of V-PCA and OA-V-PA in terms of the CDF of the location error. We observe that V-PCA performs better than OA-V-PA. The performance of the two V-PCA schemes are close, and the performance of V-PCA with two circles is slightly better than that of V-PCA with a circle and a semicircle. V-PCA schemes are able to achieve a 97th percentile accuracy of about 10 cm, which is slightly better than the performance of V-PA in Fig. \ref{CDF-3types}. This is because OA-V-PA is also adaptively selected together with V-PCA in Fig. \ref{CDF-3types}. As shown in Fig. \ref{CDF-3types}, OA-V-PA with two semicircles can achieve an 86th percentile accuracy of about 10 cm, while OV-V-PA with two superior arcs can only achieve a 50th percentile accuracy of about 10 cm. Therefore, the location accuracy improves as the available arc length of captured luminaire increases.

\begin{figure}[t]
\centering
\subfigure[]{\includegraphics[height=6cm,trim=0 0 0 0,clip]{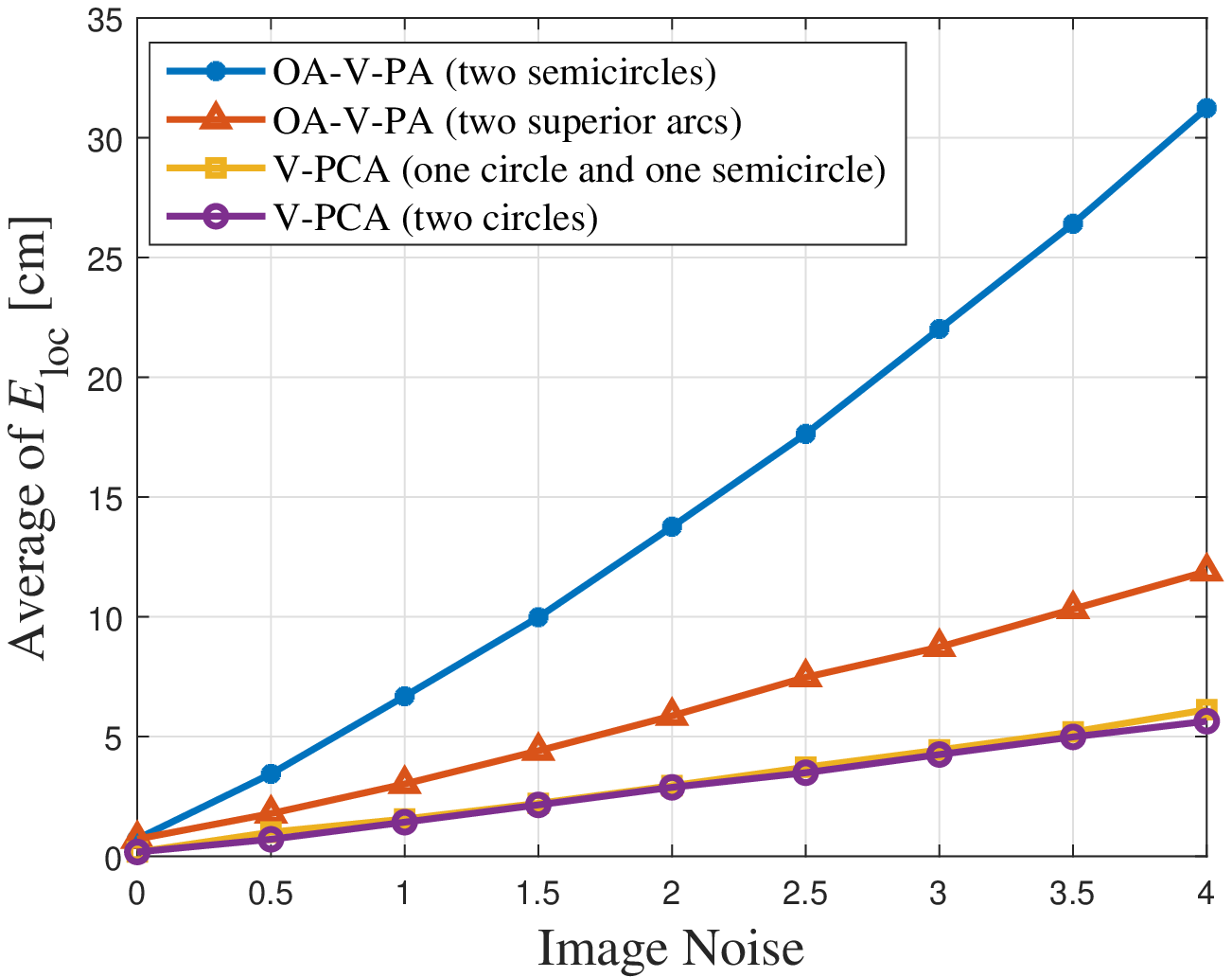}}
\subfigure[]{\includegraphics[height=6cm,trim=0 0 0 0,clip]{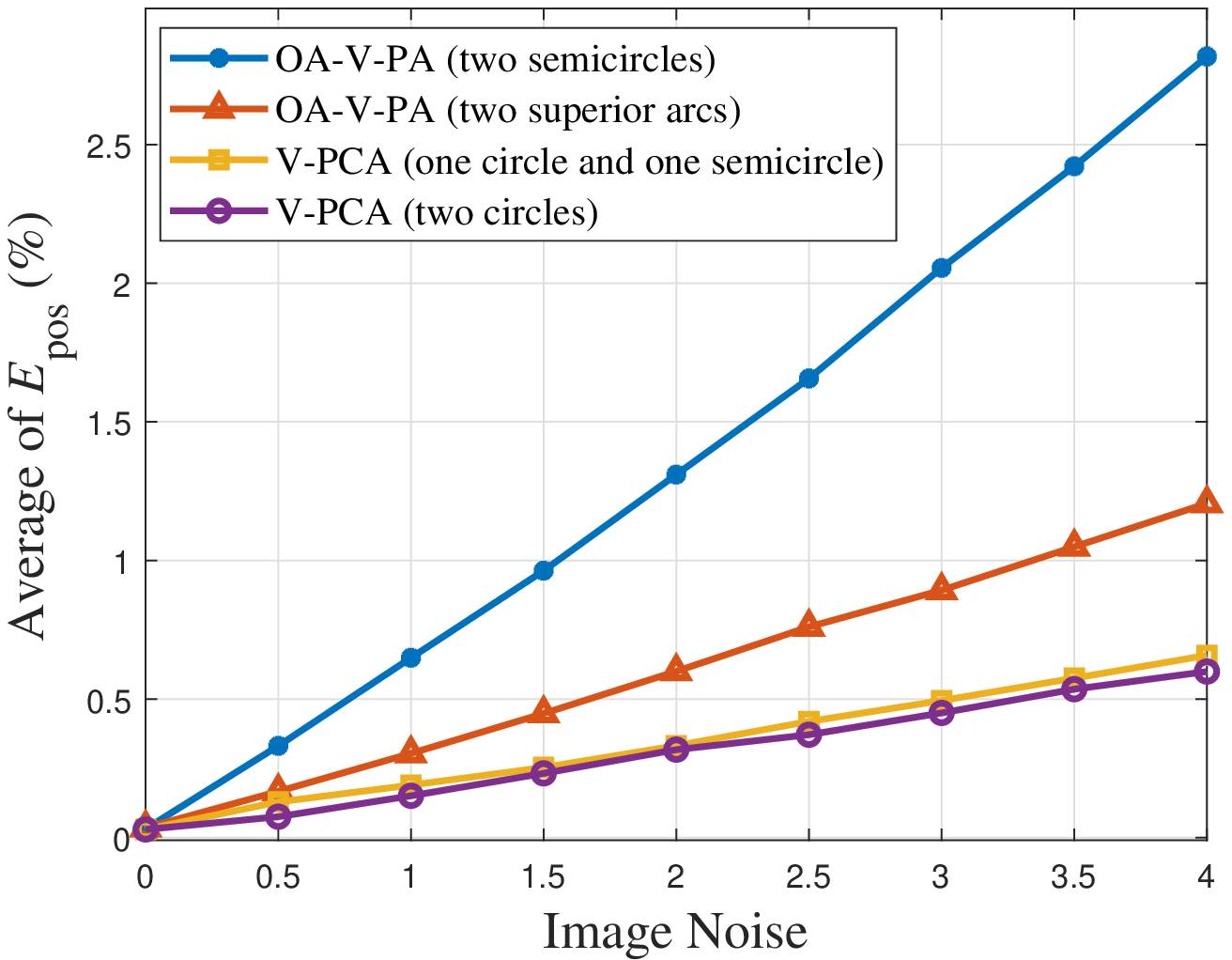}}
\caption{The mean of $E_{\rm{loc}}$ and $E_{\rm{pos}}$ versus the image noise. }\label{imagenoise-diffarc}
\end{figure}

\begin{figure}[t]
\centering
\subfigure[]{\includegraphics[height=6cm,trim=5 0 10 10,clip]{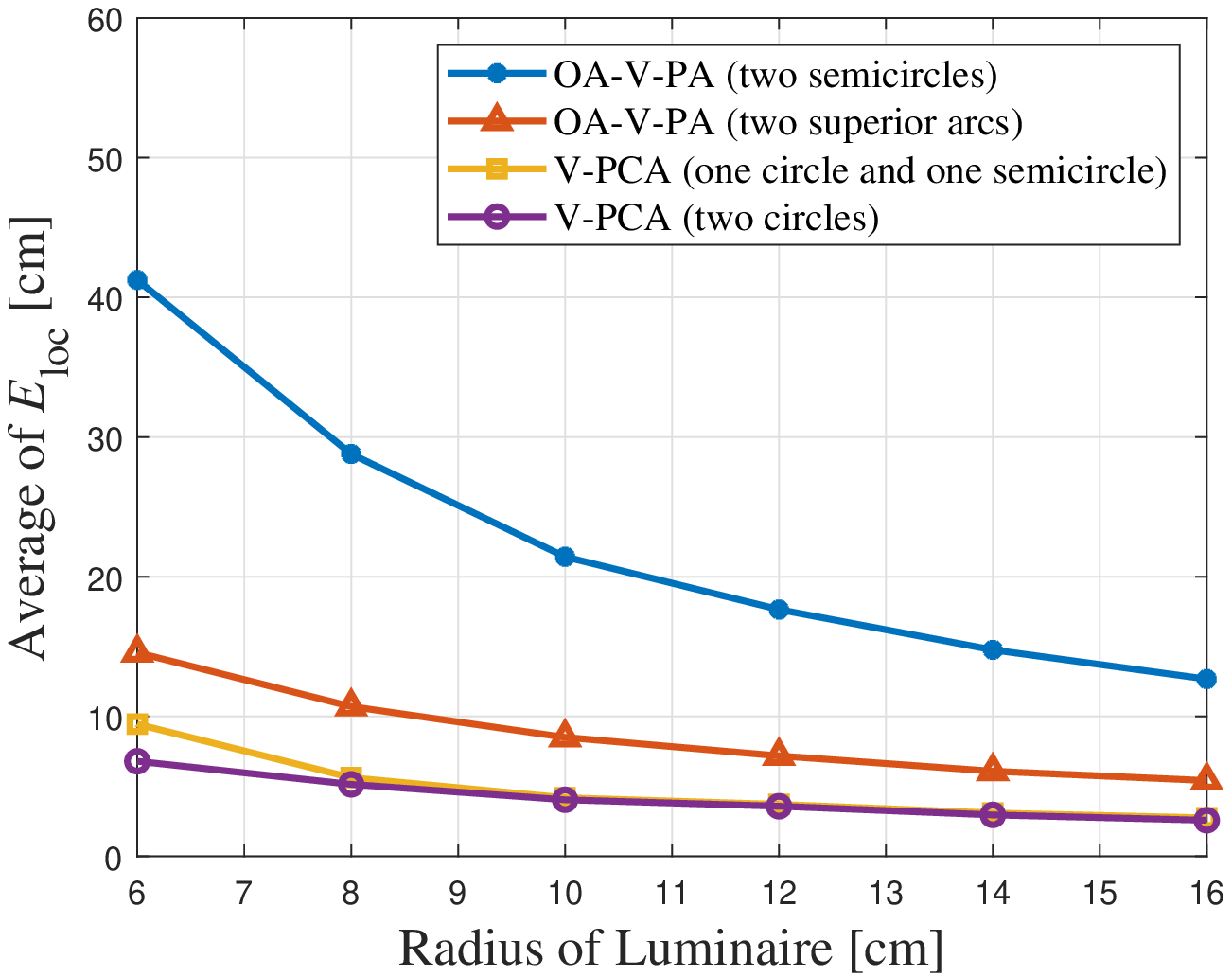}}
\subfigure[]{\includegraphics[height=6cm,trim=6 0 10 6,clip]{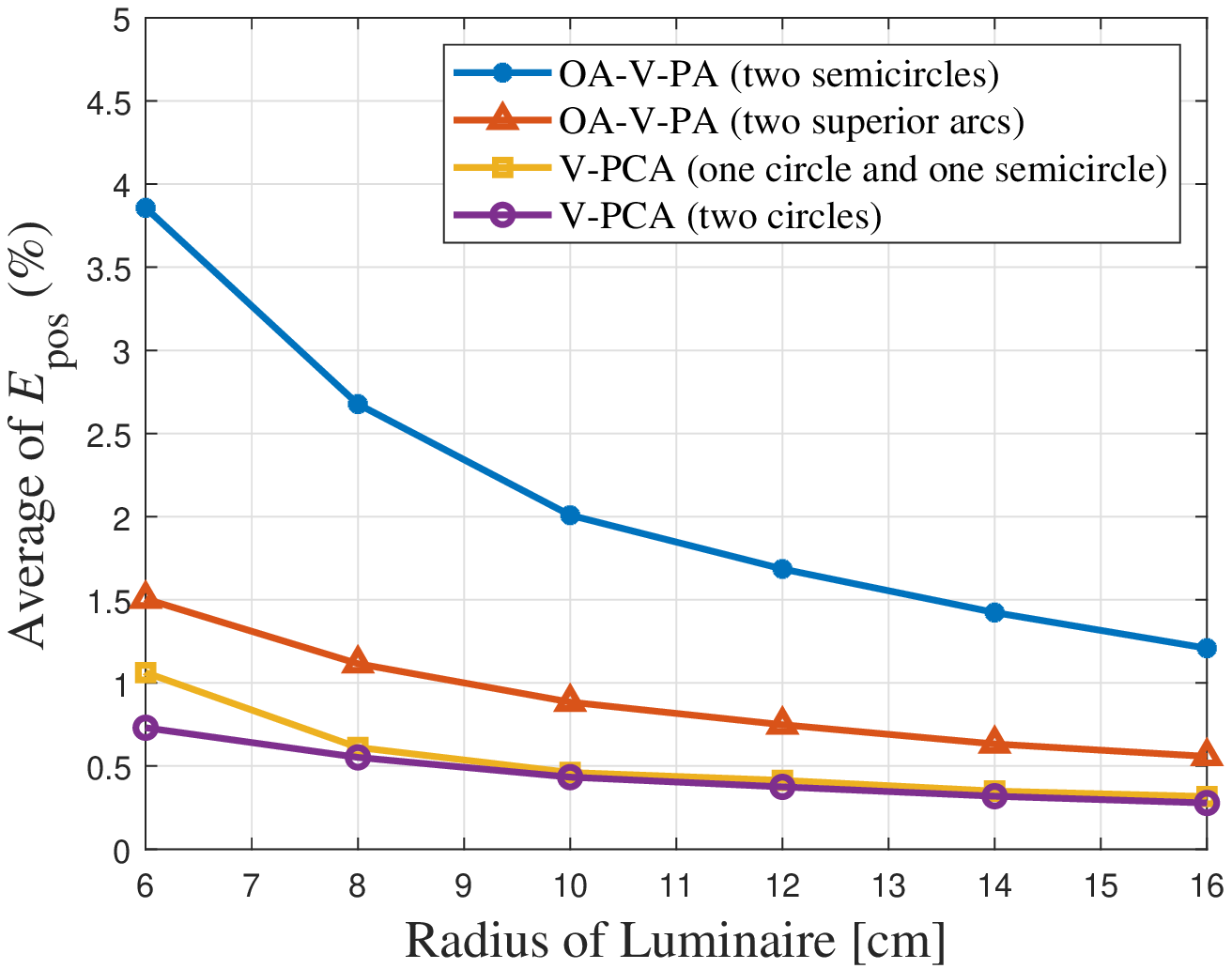}}
\caption{The mean of $E_{\rm{loc}}$ and $E_{\rm{pos}}$ versus the radius of the luminaire. }\label{radius-diffarc}
\end{figure}

Figure \ref{imagenoise-diffarc} shows the effect of the arc length of the extracted contour on positioning accuracy as image noise increases. The image noise is modeled with a standard deviation ${\sigma _n}$ that ranges from 0 to 4 pixels. From Fig. \ref{imagenoise-diffarc}(a), we can observe that the average $E_\textrm{loc}$ is 0 cm when ${\sigma _n}$ is zero for V-PCA, while it is about 0.7 cm for OA-V-PA. This is because, in OA-V-PA, the projection of the luminaires center is an approximative value, which leads to slight deviations. Moreover, the average $E_\textrm{loc}$ of OA-V-PA with two semicircles increases from 0.7 cm to 31.2 cm as the image noise increases from 0 to 4 pixels, while that of OA-V-PA with two superior arcs increases from 0.7 cm to 11.8 cm. For V-PCA with a circle and a semicircle, the average $E_{loc}$ increases from 0 cm to 5.6 cm, while it increases from 0 cm to 6.1 cm for V-PCA with two circles. Fig. \ref{imagenoise-diffarc} (b) further illustrates the average $E_\textrm{pos}$ versus the image noise.
Here, we can see that the average $E_\textrm{pos}$ of all schemes increase as the image noise increases. Both Fig. \ref{imagenoise-diffarc} (a) and Fig. \ref{imagenoise-diffarc} (b) indicate that the robustness to image noise increases when the arc length of captured luminaire increases. In the situation in which a complete circle is captured, this change is not obvious whether the other captured luminaire is complete or incomplete. This is due to the fact that the other captured luminaire is only used to eliminate the duality of the normal vector but not estimate the location of the user.

Figure \ref{radius-diffarc} evaluates the effect of the arc length of the extracted contour on positioning accuracy as the radius of the luminaire increases. The radius varies from 6 cm to 16 cm. As shown in Fig. \ref{radius-diffarc}, V-PCA with two circles has the best performance in terms of the average $E_\textrm{loc}$ and the average $E_\textrm{pos}$. In Fig. \ref{radius-diffarc} (a), the average $E_\textrm{loc}$ of OA-V-PA with two semicircles decreases from 41.2 cm to 12.7cm, while the average $E_\textrm{loc}$ of OA-V-PA with two superior arcs decreases from 14.6 cm to 5.4 cm. We can also observe that the averages of $E_\textrm{loc}$ of two V-PCA schemes are both below 10 cm. 
Fig. \ref{radius-diffarc} (b) compares the average $E_\textrm{pos}$ of OV-V-PA and V-PCA versus the radius of luminaire. We can observe that V-PCA performs better than OA-V-PA. This also verifies that the positioning performance increases as the arc length of the captured luminaire increases. Although the accuracy of OA-V-PA is slightly lower than V-PCA, it improves as the arc length of the captured luminaire increases.

\subsection{Experimental Results}

\begin{figure}[pbht]
\centering
\subfigure[The user is tilted ${{{0}^ \circ }}$ ]{\includegraphics[height=6cm,trim=0 0 0 0,clip]{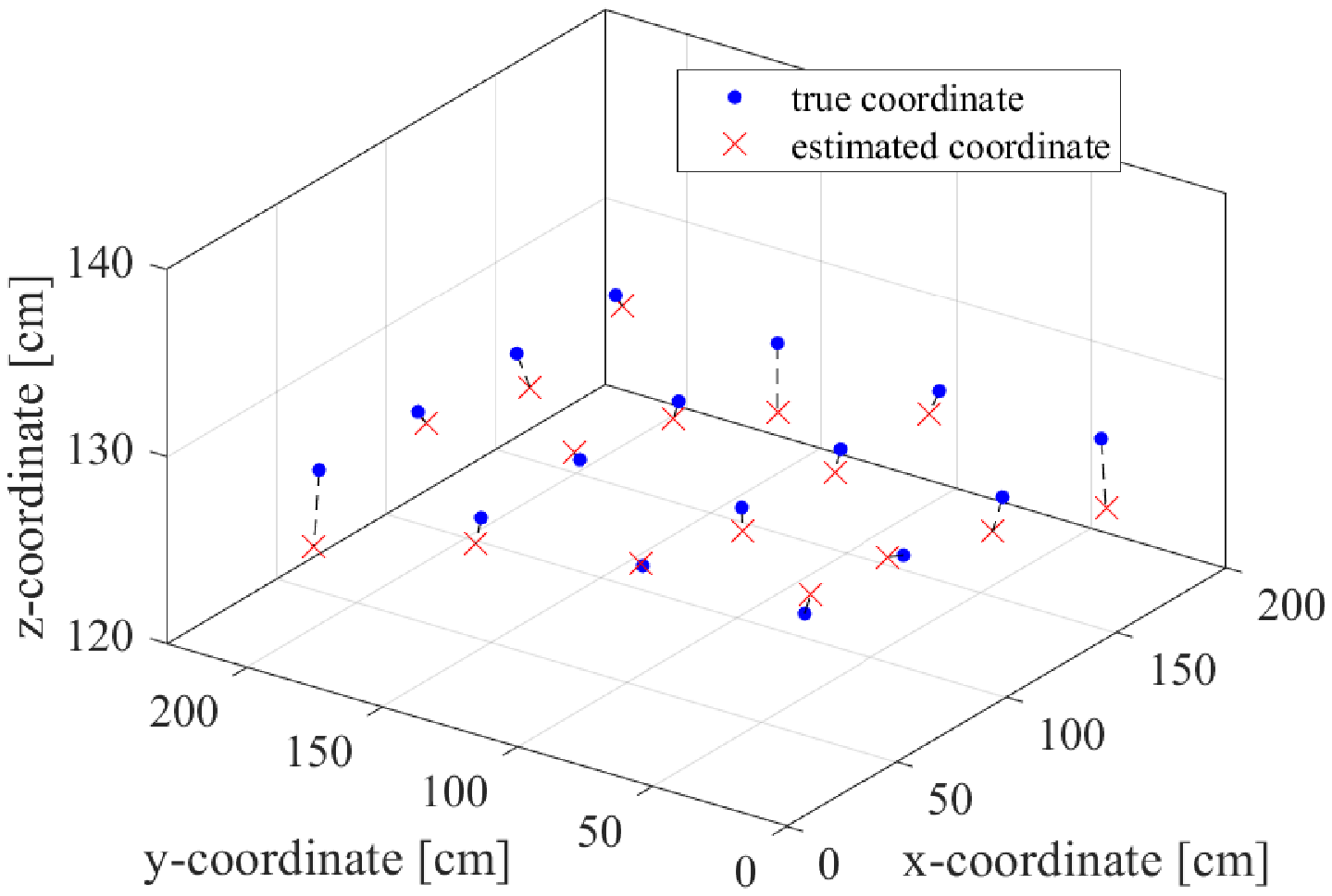}}
\subfigure[The user is tilted ${{{15}^ \circ }}$ ]{\includegraphics[height=6cm,trim=0 0 0 0,clip]{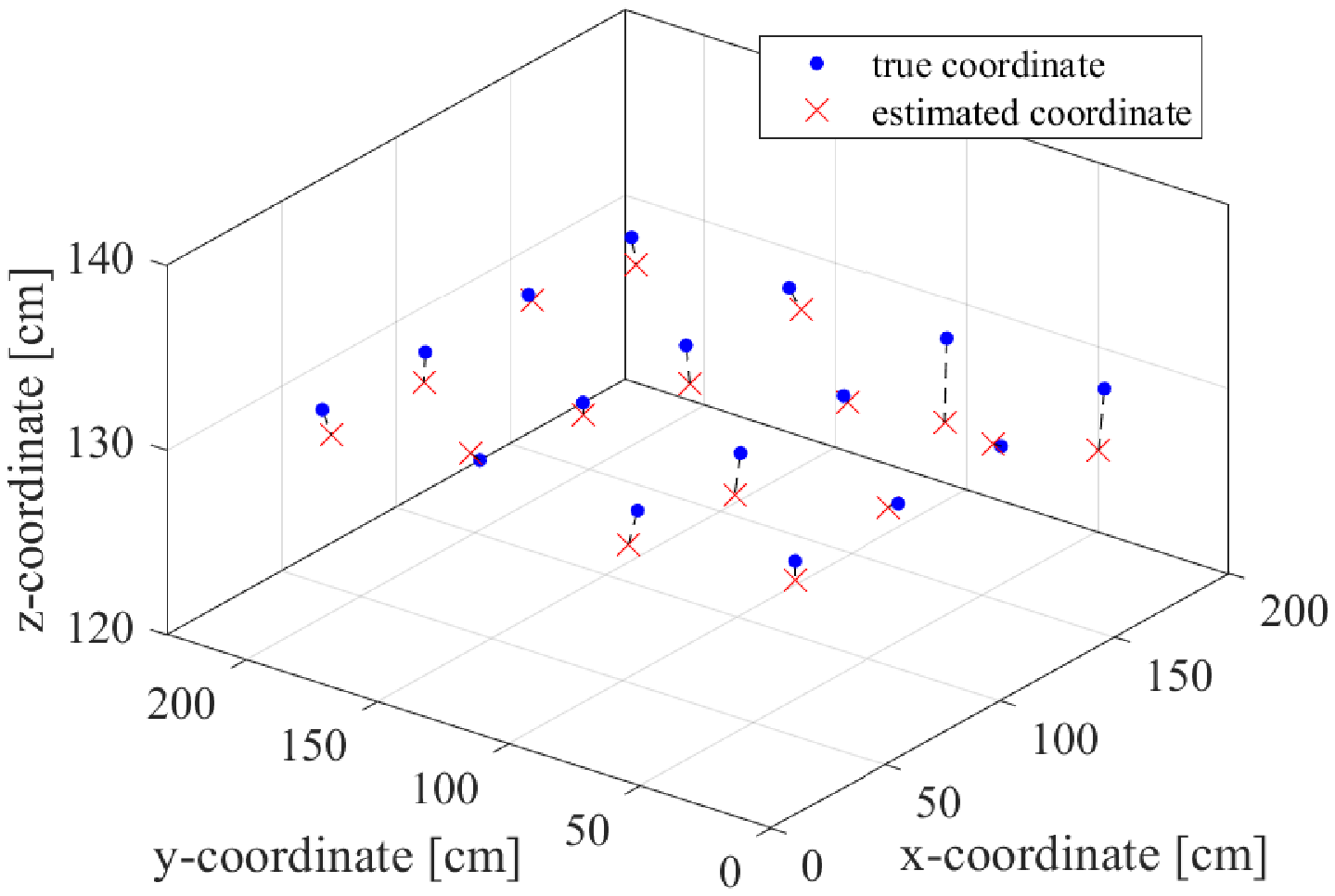}}
\subfigure[The user is tilted ${{{30}^ \circ }}$ ]{\includegraphics[height=6cm,trim=0 0 0 0,clip]{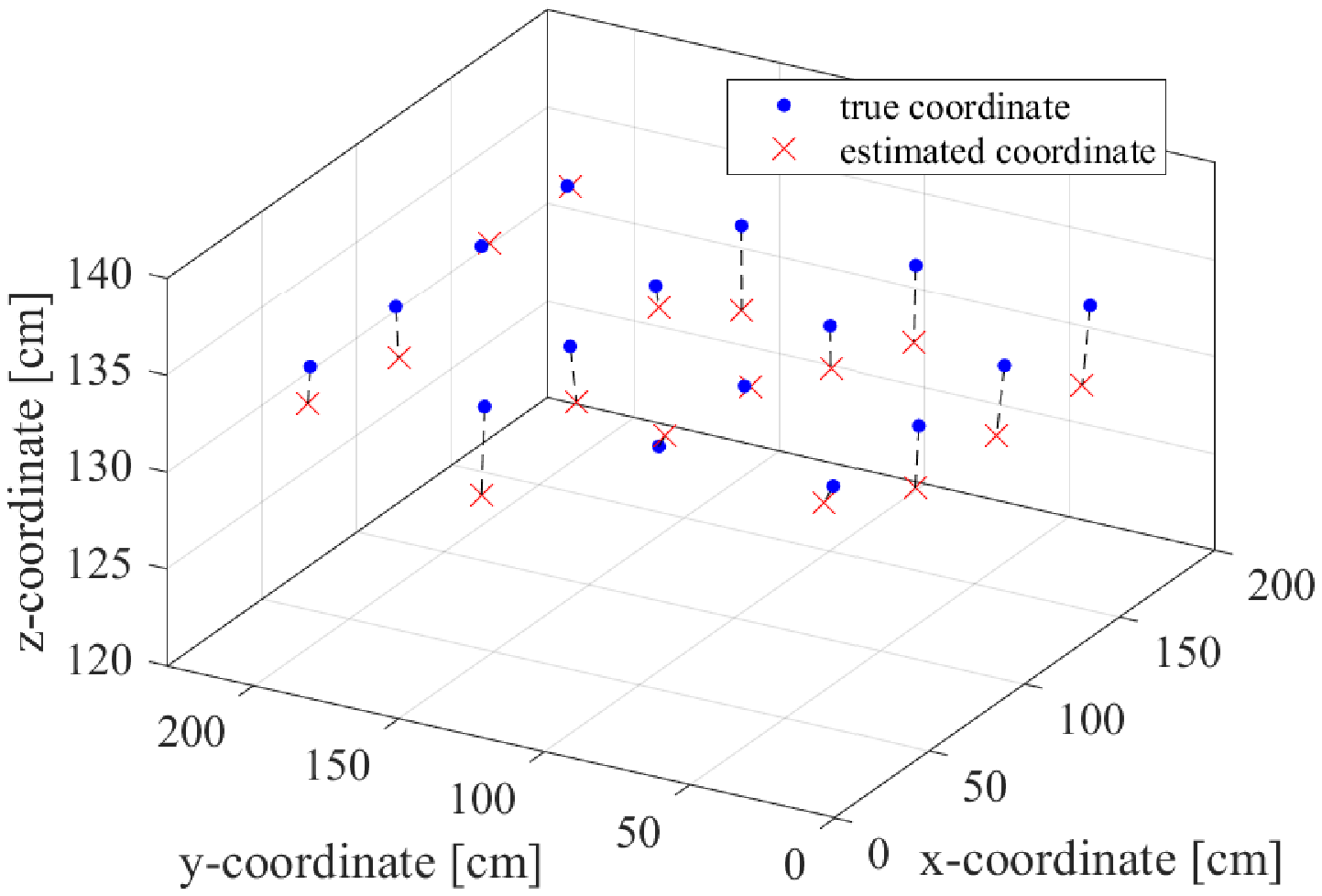}}
\caption{The accuracy performance for 3D positioning of V-PA.}\label{experimental_3D}
\end{figure}

This subsection evaluates the accuracy of V-PA via practical experiments. To verify the performance of V-PA, we estimate the locations of 16 points uniformly distributed in a 180 $\rm{cm}$ $\times$ 210 $\rm{cm}$ test area. The coordinates (in cm) of the reference points are (45, 30), (45, 90), (45, 150), (45, 210), (90, 30), (90, 90), (90, 150), (90, 210), (135, 30), (135, 90), (135, 150), (135, 210), (180, 30), (180, 90), (180, 150) and (180, 210). We evaluate 3D positioning accuracy of V-PA when the user is tilted or faces vertically upwards. For each test point, we locate the user with V-PA for 10 times, and average them as the final estimation result. Fig. \ref{experimental_3D} illustrates the positioning results of the user at different tilted angles, and the length of the dotted line that connects the true location and the corresponding estimated location denotes the magnitude of the positioning error. In Fig. \ref{experimental_3D} (a), the user faces vertically upwards, and the height of the user is 127 cm for all test points. From this figure, we observe that the minimum error and maximum are 0.87 cm and 6.82 cm, respectively. The average error of all test points is 3.37 cm. In Fig. \ref{experimental_3D} (b), we set the tilted angle of the user to about $15^\circ$ around $y^{\rm{w}}$-axis, and all test points are at the height of 130 cm. The minimum error, maximum error, and average error are 1.97 cm, 4.60 cm and 3.36 cm, respectively. Fig. \ref{experimental_3D} (c) shows the positioning result of V-PA when the user is tilted $30^\circ$ and its height is 133 cm. It can be observed that the minimum error, maximum error, and average error are 1.33 cm, 5.89 cm and 3.47 cm, respectively.
The averages errors under the three tilted angles are all less than 5 cm, which indicates that  V-PA is efficient when the user is tilted. Therefore, the proposed V-PA algorithm can achieve stable centimeter-level accuracy performance. The experimental results also verify the feasibility of V-PA algorithm in practice.

\section{Conclusion}
In this paper, we have proposed a practical V-PA approach for indoor positioning, which can achieve pose and location estimation of the user using circular luminaires. The proposed approach does not need IMU and has no tilted angle limitations at the user.
In particular, we have first developed an algorithm called V-PCA, in which, the geometric features extracted from a complete circular luminaire and an incomplete one have been developed for the pose and location estimation.
Furthermore, we have proposed the OA-V-PA algorithm to enhance the practicality and robustness when part VLC links are blocked.
We have also established a prototype, in which a fused image processing scheme was proposed to simultaneously obtain VLC signals and geometric information.
Simulation results show that V-PA can achieve a 90th percentile positioning accuracy of around 10 cm. The experimental results also show that the average location error is less than 5 cm whether the user is tilted or vertically upward. Therefore, V-PA is promising for indoor positioning, which is practical and suitable for common indoor scenarios.

\bibliographystyle{IEEEtran}
\bibliography{Citation_VLP}

\end{document}